\documentclass[journal]{IEEEtran}

\usepackage{graphicx}  % Written by David Carlisle and Sebastian Rahtz

\usepackage{amsmath}   % From the American Mathematical Society
\usepackage{amsfonts}

\usepackage{latexsym}

\newcommand{\mean}[1]{\langle{#1}\rangle}
\newcommand{\pro}[2]{\langle{#1}|{#2}\rangle}
\newcommand{\bra}[1]{\langle{#1}|}
\newcommand{\ket}[1]{|{#1}\rangle}

\newcommand{\dgg}{^{\dagger}}

\newcommand{\half}{\frac{1}{2}}

\newcommand{\veca}{\mbox{\boldmath $a$}}
\newcommand{\vecb}{\mbox{\boldmath $b$}}
\newcommand{\vecx}{\mbox{\boldmath $x$}}
\newcommand{\vecy}{\mbox{\boldmath $y$}}
\newcommand{\vecu}{\mbox{\boldmath $u$}}

\newtheorem{theorem}{Theorem}[section]
\newtheorem{corollary}{Corollary}[section]
\newtheorem{lemma}{Lemma}[section]
\newtheorem{definition}{Definition}[section]
\newtheorem{example}{Example}[section]

\begin{document}

\title{
System identification for \\ passive linear quantum systems
%System identification for linear quantum systems
}


%%%%%%%%%%%%%%%%%%%%%%%%%%%%%%%%%%%%%%%%%%%
%%%%%%%%%%%%%%%%%%%%%%%%%%%%%%%%%%%%%%%%%%%
%%%%%%%%%%%%%%%%%%%%%%%%%%%%%%%%%%%%%%%%%%%

% For single-column version

%\author{
%M\u{a}d\u{a}lin Gu\c{t}\u{a}$\mbox{}^{1}$ 
%and 
%Naoki Yamamoto$\mbox{}^{2}$
%}

% For double-column version

\author{M\u{a}d\u{a}lin Gu\c{t}\u{a} and Naoki Yamamoto
\thanks{M. Gu\c{t}\u{a} is with the School of Mathematical Sciences, University 
        of Nottingham, 
        University Park, NG7 2RD Nottingham, UK 
        (email: madalin.guta@nottingham.ac.uk). 
        N. Yamamoto is with the Department of Applied Physics and 
        Physico-Informatics, Keio University, 
        Hiyoshi 3-14-1, Kohoku, Yokohama 223-8522, Japan 
        (e-mail: yamamoto@appi.keio.ac.jp). }
%\thanks{
%This work was partly supported by JSPS of Japan.
%}
}

\maketitle

%% For the one-column submitted version

%\vspace{-1.5cm}
%
%{\small
%1. School of Mathematical Sciences, University of Nottingham, 
%University Park, NG7 2RD Nottingham, UK. 
%Phone: +44 (0) 115 951 4993, 
%Fax: +44 (0) 115 951 4951, 
%Email: madalin.guta@nottingham.ac.uk
%}
%
%{\small
%2. Department of Applied Physics and Physico-Informatics, 
%Keio University, Hiyoshi 3-14-1, Kohoku, Yokohama 223-8522, Japan. 
%Phone: +81 45 566 1830, 
%FAX: +81 45 566 1587, 
%Email: yamamoto@appi.keio.ac.jp
%}
%
%
%\vspace{0.5cm}

%%%%%%%%%%%%%%%%%%%%%%%%%%%%%%%%%%%%%%%%%%%
%%%%%%%%%%%%%%%%%%%%%%%%%%%%%%%%%%%%%%%%%%%
%%%%%%%%%%%%%%%%%%%%%%%%%%%%%%%%%%%%%%%%%%%

\begin{abstract}

System identification is a key enabling component for the implementation 
of quantum technologies, including quantum control. 
In this paper, we consider the class of passive linear input-output systems, 
and investigate several basic questions: 
(1)  which parameters  can be identified? 
(2) Given sufficient input-output data, how do we reconstruct the system 
parameters? 
(3) How can we optimize the estimation precision by preparing appropriate 
input states and performing measurements on the output? 
We show that minimal systems can be identified up to a unitary 
transformation on the modes, and systems satisfying a Hamiltonian connectivity 
condition called ``infecting'' are completely identifiable. 
We propose a frequency domain design based on a Fisher information 
criterion, for optimizing the estimation precision for coherent input state. 
As a consequence of the unitarity of the transfer function, we show that the 
Heisenberg limit with respect to the input energy can be achieved using 
non-classical input states.

\end{abstract}


%\begin{center} \bfseries EDICS Category: 3-BBND \end{center}
% Will be ignored for other modes.

\IEEEpeerreviewmaketitle

\section{Introduction}

We are currently witnessing the beginning of a quantum engineering 
revolution \cite{Dowling&Milburn}, marking a shift from ``classical 
devices'' which are macroscopic systems described by deterministic or 
stochastic equations, to ``quantum devices''  which exploit fundamental 
properties of quantum mechanics, with applications ranging from 
computation to secure communication and metrology 
\cite{NielsenBook,Furusawa-2011}. 
While control theory was developed from the need for predictability in the 
behavior of ``classical'' dynamical systems, quantum filtering and quantum 
feedback control theory 
\cite{Belavkin1999,Mabuchi&Khaneja,James&Bouten&vanHandel} deal 
with similar questions in the mathematical framework of quantum 
dynamical systems.

System identification is an essential component of control theory, which 
deals with the estimation of unknown dynamical parameters of input-output 
systems; 
in particular, the identification of \emph{linear} systems is a well studied 
subject in classical systems theory \cite{LjungBook}. 
A similar task arises in the quantum setup, and various aspects of the 
\emph{quantum system identification} problem have been considered 
in the recent literature, cf. 
\cite{Mabuchi1996,Gambetta2001,Stockton2004,Chase2009,Burgarth2009PRA,
Burgarth2009NJP,Schirmer&Oi,Burgarth2011NJP,KatoYamamoto} for a 
shortlist of recent results. 
Further, detailed statistical analysis for some dynamical quantum 
identification problems have been demonstrated 
\cite{Guta2011,Catana&vanHorssen&Guta,Guta&Kiukas,Catana&Bouten&Guta}.

In this paper, we focus on the class of \emph{passive linear quantum system} 
\cite{Gough2008,Nurdin2010,Petersen2011,WisemanBook}, which serves as 
a device for several applications in quantum information technology, 
such as entanglement generation 
\cite{Parkins,Ficek2009,Krauter,Muschik,Yamamoto2012}, quantum memory 
\cite{Gorshkov,Simon,Majer,Hush,He2009,Yamamoto2014}, and linear 
quantum computing \cite{KLM}. 
Analyzing this important class of systems provides the foundation for the general case, 
but it has a clear interest in its own right in the context of estimation, as described later
in this section. 
The system consists of a number of quantum variables (e.g. the electromagnetic 
field inside an optical cavity), and is coupled with the quantum stochastic input 
consisting of non-commuting noise processes (e.g. a laser impinging onto the 
cavity mirror). 
As a result of the quantum mechanical interaction between system and input, 
the latter is transformed into an output quantum signal which can be measured 
to produce a classical stochastic measurement process. 
In this context, we address the problem of identifying the linear system by 
appropriately choosing the state of its input and performing measurements 
on the output (see Fig.~\ref{Sys ID setup}).

\begin{figure}
\centering
\includegraphics[scale=0.34]{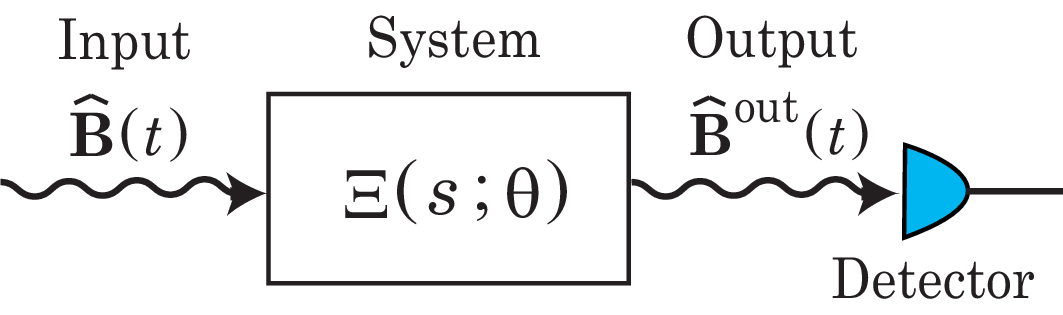}
\caption{\label{Sys ID setup}
Setup of system identification for linear quantum systems. 
The experimenter can prepare a time-dependent input state, and perform 
a continuous-time measurement on the output, from which the unknown 
system parameters $\theta$ are estimated. 
The input-output relation is encoded in the transfer function $\Xi(s;\theta)$.
}
\end{figure}

In contrast to the classical case, a systematic methodology for linear quantum 
system identification has not yet been developed. 
Our aim is to fill this gap by investigating the following questions. 
(1) {\it Identifiability}: which system parameters can be in principle identified? 
(2) {\it Identification method}: given sufficient input-output data, how can 
we actually reconstruct system parameters? 
(3) {\it Statistics}: how well can we estimate unknown parameters by 
preparing appropriate input states and performing measurements on the 
output? 
The key fact to solve these problems is that, for linear systems, the Laplace 
domain input and output fields are related by a simple linear transformation 
represented by the \emph{transfer function matrix}.

Below we give a more detailed account of the above-mentioned problems and 
the results obtained in this paper. 
First, the system identifiability is the property guaranteeing that all the 
system parameters can be in principle uniquely determined from the 
input-output data. 
This is actually an important notion in the classical case as well 
\cite{Anderson1969,Glover1974,Grewal1976}, and recently we find some 
proposals of those quantum analogues \cite{DAlessandro2005,Burgarth2012} 
for nonlinear systems. 
In this paper, we show that minimal passive linear systems having the same 
transfer function (i.e. the equivalent class) are related by \emph{unitary} 
transformations acting on the space of modes. 
Then, based on this result, we characterize a wide class of identifiable 
quantum linear networks, by employing the concept of {\it infection} 
introduced in \cite{Burgarth2009NJP,Burgarth2011NJP}. 
Next, the problem (2) boils down to that of identifying the transfer function, 
which can then be used to reconstruct the parameters of the system; 
in our case, those are the system's (quadratic) Hamiltonian and its coupling 
to the environment, both described by appropriate matrices. 
In this paper, we provide two methods for finding the identifiable parameters 
and physical realizations for a given transfer function.

Beyond identifiability, it is important to investigate and compare the 
\emph{statistical performance} of different estimation methods. 
By employing the well-established quantum estimation theory 
\cite{Holevo,Braunstein&Caves}, in particular the notion of 
\emph{quantum Fisher information}, we investigate the problem of 
devising optimal (time dependent) coherent input states of a given energy, 
and output measurements. 
%For one dimensional parameter we show that the optimal strategy is to probe 
%the system at the frequency where the derivative of the transfer function 
%has the largest norm. 
More precisely, we study the special case of a single-mode, single-input 
single-output (SISO) system in several scenarios with one or two unknown 
parameters. 
Moreover, for the single-mode SISO system, we show that the {\it Heisenberg 
limit} with respect to the input energy can be achieved for a non-classical input 
state. 
%As in the standard quantum metrology \cite{Giovanetti}, this can be traced 
%back to the unitarity of the transfer function as input-output transform. 
Note that, although this enhanced statistical performance could be expected 
from the quantum metrology theory \cite{Giovanetti}, the important new 
concept is that this is the metrology for a {\it dynamical} system, where the 
static phase is now replaced by a dynamical phase represented by the transfer 
function. 
In fact this setup poses some new problems; for instance we need to optimize 
the frequency of the input field, which is not considered in the standard 
quantum metrology dealing with only static parameter estimation problems. 
These new problems can be formulated and solved thanks to 
the unitarity of the transfer function of linear passive systems, which is one 
of the reasons why we are chosen to investigate this class of 
systems separately from more general, active linear systems.

For reader's convenience  we summarize in advance the new concepts appearing in the quantum 
system identification problems studied in this paper, which are not found in 
the conventional identification theory for classical systems. 
The system's input-output relation is represented 
by a transfer function having a special structure, which stems from the joint 
unitary evolution of the system and the field, and the fact that the interaction is passive. 
As consequence, the equivalence classes of parameters with the same output 
can be characterized in terms of unitary, rather than 
a general invertible matrices as is the case for classical systems. 
Note that limiting to a special class of linear systems does not mean 
straightforward applicability of the general identification theory for classical 
systems, but we need to take into account the essential feature of the focused 
system. 
Another specifically quantum aspect of the present theory is that all our 
results apply also to non-classical input states such as a single photon field; 
indeed, the transfer function can be used to describe the input-output relation 
even in such strong quantum scenarios \cite{Guofeng2013}, which is one of 
the advantages of the linear setup.
%this is because, while those results are obtained based only on the transfer 
%function, the input-output relation even in such strong quantum scenarios 
%can still be represented by the transfer function \cite{Guofeng2013}, which 
%is indeed one of the strong advantages available only in the linear case. 
This fact is important for the following two reasons. 
First, as mentioned in the above paragraph, the enhanced quantum 
system identification is achieved for non-classical input states. 
Second, such a passive linear systems driven by single photons 
behave essentially in the same way as some nonlinear/finite-level systems 
such as a dissipative qubit network driven by a single photon \cite{Yu Pan 2015}; 
hence the theory developed in this paper is applicable to those genuine 
quantum systems beyond linear regime.

The paper is structured as follows. 
In Section~II we introduce the setup of passive linear quantum systems, 
illustrated with realistic examples of system identification problems. 
In Section~III, we give a necessary and sufficient condition for the 
identifiability of a passive linear system, which is then applied to several 
examples. 
Section~IV describes the class of infective networks, which are shown to 
be completely identifiable. 
Section~V provides two concrete identification methods. 
Section~VI is devoted to the statistical analysis of the identification problem, 
using a Fisher information approach for the optimization over input states 
and output measurements. 
In Section~VII, we briefly discuss the case of general (i.e. active) systems, 
pointing out some similarities and differences from the passive case, and 
formulate a conjecture regarding the structure of the equivalence classes.

Throughout the paper we will use the following notations: 
for a matrix $A=(a_{ij})$, the symbols $A^\dagger$ and $A^T$ 
represent its Hermitian conjugate and transpose of $A$, i.e., 
$A^\dagger=(a_{ji}^*)$ and $A^T=(a_{ji})$, respectively. 
For a matrix of operators, $\hat A=(\hat a_{ij})$, we use the 
same notation, in which case $\hat a_{ij}^*$ denotes the adjoint 
to $\hat a_{ij}$. 
$I_n$ denotes the $n\times n$ identity matrix.

A preliminary version of this paper was presented 
at the 52nd IEEE CDC \cite{GutaYamamotoCDC}.

%%%%%%%%%%%%%%%%%%%%%%%%%%%%%%%%%%%%%%%%%%%
%%%%%%%%%%%%%%%%%%%%%%%%%%%%%%%%%%%%%%%%%%%
%%%%%%%%%%%%%%%%%%%%%%%%%%%%%%%%%%%%%%%%%%%

\section{Passive linear quantum systems}

In this section we briefly review the framework of linear classical and quantum 
dynamical systems, with several examples showing the need of 
system identification.

\subsection{Classical linear systems}

A classical linear system is described by the set of differential equations
\[    
    d\vecx(t) = A\vecx(t)dt+B\vecu(t)dt,~~~
    d\vecy(t) = C\vecx(t)dt + D\vecu(t)dt, 
\]
%
%% For double-column format:
%
%\begin{eqnarray*}
%d\vecx(t)&=&A\vecx(t)dt+B\vecu(t)dt,\\
%d\vecy(t)&=&C\vecx(t)dt + D\vecu(t)dt, 
%\end{eqnarray*} 
%
where $\vecx(t)\in \mathbb{R}^{n}$ is the state of the system, 
$\vecu(t)\in \mathbb{R}^{m}$ is an input signal, and 
$\vecy(t)\in \mathbb{R}^{k}$ is the output signal. 
The observer can control the input signal and observe the output, but does 
not have access to the internal state of the system. 
The input signal can be deterministic, in which case we deal with a set of 
ODEs, or stochastic, in which case the equations should be interpreted as SDEs. 
Apart from the input and the initial state of the system, the dynamics is 
determined by the (real) matrices $A,B,C,D$.

To find the relation between input and output it is convenient to work in 
the Laplace domain. 
The Laplace transform of $\vecx(t)$ is defined by 
\begin{equation}\label{eq.laplace.classic}
    {\cal L}[\vecx](s)
      :=\int_0^\infty e^{-st}\vecx(t)dt, 
\end{equation}
where ${\rm Re}(s)>0$. 
Then, we have the explicit input-output relation 
${\cal L}[\vecy](s) =  \Xi(s){\cal L}[\vecu](s)$, where 
\begin{equation}
\label{eq.transfer.function.classic}
    \Xi(s) = C(sI - A)^{-1}B+ D
\end{equation}
is the \emph{transfer function matrix}. 
System identification deals with the problem of estimating the matrices 
$A,B,C,D$ or certain parameters on which they depend, from the knowledge 
of the input and output processes. 
From \eqref{eq.transfer.function.classic} it is clear that the observer can 
at most determine the transfer function $\Xi(s)$ by preparing appropriate 
inputs and observing the output.

The identifiability problem is closely related to the fundamental system 
theory concepts of \emph{controllability} and \emph{observability}. 
The system is controllable if for any states $\vecx_{0}, \vecx_{1}$ and 
times $t_{0}<t_{1}$ there exists a (piece-wise continuous) input $\vecu(t)$ 
such that the initial and final states are given by $\vecx(t_{0})= \vecx_{0}$ 
and $\vecx(t_{1}) = \vecx_{1}$, respectively. 
This is equivalent to the fact that the controllability matrix 
${\cal C}=[B, AB, \ldots, A^{n-1} B]$ has full row rank. 
The system is observable if for any times $t_{0}<t_{1}$, the initial state 
$\vecx(t_{0})= \vecx_{0}$ can be determined from the history of the input and 
output on the time interval $[t_{0},t_{1}]$. 
This is in turn equivalent to the fact that the observability matrix 
${\cal O}=[C^{T}, (CA)^{T}, \ldots, (CA^{n-1})^{T}]^{T}$ has full column rank.

The importance of these concepts for identifiability stems from the fact that 
if the system is \emph{not} controllable or observable then there exists 
a lower dimensional system with the same transfer function as the original 
one. 
The former  can be obtained from the latter by separating its coordinates 
via a canonical procedure called the Kalman decomposition. 
Therefore, in system identification it is natural to restrict the attention 
to \emph{minimal} systems, i.e. systems which are both controllable and 
observable. 
As noted above, by appropriately choosing the input signal $\vecu(t)$, the 
observer can effectively identify the transfer function $\Xi(s)$, while other 
independent parameters in the system matrices are not identifiable in the 
absence of any prior knowledge. 
The following theorem gives a precise characterization of systems which 
are equivalent in the sense that they cannot be distinguished based on 
the input-output history \cite{LjungBook}.

\begin{theorem} 
\label{th.equivalence.class} 
Two minimal systems $(A,B,C,D)$ and $(A^{\prime} , B^{\prime}, C^{\prime}, 
D^{\prime})$ have the same transfer function $\Xi(s)$ if and only if they 
are related by a similarity transformation,  i.e. there exists an invertible 
$n\times n$ matrix $T$ such that
$$
      A^{\prime} = TAT^{-1}, \quad B^{\prime} = TB, 
      \quad C^{\prime} = CT^{-1}, \quad D^{\prime}= D.
$$
\end{theorem}

\subsection{Passive linear quantum system}

A general linear quantum system with $n$ continuous variables 
modes is described by the column vectors of creation operators 
$\hat{\veca}^{*}:=[\hat{a}^*_1,\ldots,\hat{a}^*_n]^{T}$ and 
annihilation operators $\hat{\veca}:=[\hat{a}_1,\ldots,\hat{a}_n]^{T}$ 
satisfying the commutation relations 
\begin{equation}
\label{CCR}
     \hat{a}_{i} \hat{a}^{*}_{j} - \hat{a}^{*}_{j}\hat{a}_{i} 
         = [\hat{a}_{i},\hat{a}^{*}_{j}] 
         %= \delta_{ij}\mathbf{1}.
         = \delta_{ij}\hat{1}.
\end{equation}
The system has a quadratic Hamiltonian of the form
\[
    \hat{H} = \hat{\veca}^{\dagger} \Omega \hat{\veca}
     = [\hat{a}^*_1,\ldots,\hat{a}^*_n]
      \left[ \begin{array}{ccc}
        \Omega_{11} & \ldots & \Omega_{1n} \\
        \vdots      &        & \vdots      \\
        \Omega_{n1} & \ldots & \Omega_{nn} \\
      \end{array} \right]
      \left[ \begin{array}{c}
       \hat{a}_1 \\
       \vdots    \\
       \hat{a}_n \\
     \end{array} \right]
\]
%
%% For double-column format:
%
%\begin{eqnarray*}
%    \hat{H}  &=& \hat{\veca}^{\dagger} \Omega \hat{\veca}
%    \nonumber\\
%    &=&[\hat{a}^*_1,\ldots,\hat{a}^*_n]
%      \left[ \begin{array}{ccc}
%        \Omega_{11} & \ldots & \Omega_{1n} \\
%        \vdots      &        & \vdots      \\
%        \Omega_{n1} & \ldots & \Omega_{nn} \\
%      \end{array} \right]
%      \left[ \begin{array}{c}
%       \hat{a}_1 \\
%       \vdots    \\
%       \hat{a}_n \\
%     \end{array} \right]
%\end{eqnarray*}
%
with $\Omega$ an $n\times n$ complex Hermitian matrix, and is coupled to 
$m$ bosonic quantum fields 
$\hat{\bf B}(t)= [\hat{B}_{1}(t),\dots, \hat{B}_{m}(t)]^T$ whose 
algebraic properties are characterized by the commutation 
relations 
$$
     [\hat B_{i}(t), \hat B_{j}^{*}(s) ] 
       %= \min\{s,t\} \cdot \delta_{ij} \cdot\mathbf{1},
       = \min\{s,t\} \delta_{ij} \hat{1},
$$
or alternatively by 
\begin{equation}
\label{eq.comm.relations}
     %[\hat{\vecb}(t), \hat{\vecb}^{*}(s) ]  =\delta(t-s) \cdot  \mathbf{1}.
     [\hat b_i(t), \hat b_j^{*}(s) ] = \delta(t-s) \delta_{ij}\hat{1}.
\end{equation}
where $\hat{\vecb}(t)=[\hat b_1(t),\ldots,\hat b_m(t)]^T$ is the white 
noise operator formally defined as $\hat{\vecb}(t) = d\hat{\bf B}(t)/dt$.

The coupling between system and field is described by the 
following set of operators: 
\begin{equation*}
     \hat{{\bf L}} = C\hat{\veca}
      = \left[ \begin{array}{ccc}
          c_{11} & \ldots & c_{1n} \\
          \vdots &        & \vdots \\
          c_{m1} & \ldots & c_{mn} \\
        \end{array} \right]
        \left[ \begin{array}{c}
          \hat{a}_1 \\
          \vdots    \\
          \hat{a}_n \\
        \end{array} \right],
\end{equation*}
with $c_{ij}$ a complex number. 
More precisely, the joint system-field evolution up to time 
$t$ is given by the unitary operator $\hat U(t)$ satisfying the 
quantum stochastic differential equation (QSDE) \cite{Parthasarathy}
$$
    d\hat{U}(t) 
           = \left( d\hat{\bf B}^{\dagger}(t) \hat{\bf L} 
               - \hat{\bf L}^{\dagger} d\hat{\bf B}(t) 
               + \frac{1}{2} \hat{\veca}^{\dagger} A \hat{\veca} dt 
               \right)\hat{U}(t),
$$
where
\begin{equation}
\label{A matrix}
    A:=-i\Omega-\frac{1}{2}C\dgg C. 
\end{equation}
This type of system is called ``passive", because the operators 
do not involve the creation process such as $\hat{a}^*_i\hat{a}^*_j$ 
in $\hat{H}$ and $\hat{a}^*_i$ in $\hat{\bf L}$, representing a purely 
dissipative evolution.

The Heisenberg evolution of the system operators is 
$\hat{\veca}(t)=\hat{U}(t)^*\hat{\veca} \hat{U}(t)$, which by differentiation 
gives the equation
\begin{equation}
\label{dynamics}
   d \hat{\veca}(t)=A\hat{\veca}(t) dt -C\dgg d \hat{\bf B}(t). 
\end{equation}
Similarly, the output process 
$\hat{\bf B}^{\rm out}(t)=\hat{U}(t)^*\hat{\bf B}(t)\hat{U}(t)$ 
satisfies the differential equation
\begin{equation}
\label{observation}
 d \hat{\bf B}^{\rm out}(t)=C\hat{\veca}(t) dt+d \hat{\bf B}(t). 
\end{equation}

The Laplace transforms of $\hat{\veca}(t)$, $\hat{\vecb}(t) = d\hat{\bf B}(t)/dt$, 
and $\hat{\vecb}^{\rm out}(t) = d\hat{\bf B}^{\rm out}(t)/dt$ are defined as 
in \eqref{eq.laplace.classic}, for ${\rm Re}(s)>0$. 
As we will be assuming that the system is stable, the initial state of the system 
is irrelevant in the long time limit, and we can set its mean to zero 
$\mean{\hat{\veca}(0)} =0$. 
In the Laplace domain the input-output relation is a simple multiplication 
\begin{equation}
\label{eq.input.output}
    {\cal L}[\hat{\vecb}^{\rm out}](s)
      =\Xi(s){\cal L}[\hat{\vecb}](s), 
\end{equation}
where $\Xi(s)$ is the transfer function matrix: 
\begin{equation}
\label{Q transfer function}
    \Xi(s) := I_m - C(sI-A)^{-1}C\dgg. 
\end{equation}
With $s=-i\omega$ we define the frequency domain operators
$$
\hat{\vecb}(\omega): = \mathcal{L}[\hat{\vecb}](-i\omega) = 
\frac{1}{\sqrt{2\pi}} \int_{-\infty}^\infty e^{i\omega t} \hat{\vecb}(t)
$$
so that $\hat{\vecb}^{\rm out}(\omega)=\Xi(-i\omega) \hat{\vecb}(\omega)$. 
Since $\hat{\vecb}^{\rm out}(\omega)$ must satisfy 
canonical commutation relations similar to \eqref{eq.comm.relations}, $\Xi(-i\omega)$ 
must be unitary for all $\omega$ \cite{Gough2008}. 
%%%%%%%%%%%%observability controllability 
%\subsection{Controlability and }
%The system is controllable if for any states $x_{0},x_{1}$ and times $t_{0}<t_{1}$ there exists a 
%(piece-wise continuous) input $u(t)$ such that $x(t_{0})= x_{0}$ and final state $x(t_{1})= x_{1}$. This turns out to be equivalent to the fact that the following \emph{controllability matrix} has full row rank
%\[
%   {\cal C}=[B, AB, \ldots, A^n B].
%  % {\cal O}=[C\dgg, (CA)\dgg, \ldots, (CA^n)\dgg]\dgg,
%\]
%The system is observable if for any times $t_{0}<t_{1}$, the initial state $x(t_{0})= x_{0}$ can be determined from the history of the input and output on the time interval $[t_{0},t_{1}]$. This is in turn equivalent to the fact that the 
%\emph{observability matrix} has full column rank
%\[
%  % {\cal C}=[B, AB, \ldots, A^n B].~~
%   {\cal O}=[C^{T}, (CA)^{T}, \ldots, (CA^n)^{T}]^{T}.
%\]
%

%%%%%%%%%%%%%%%%%%%%%%%%%%%%%%%%%%%%%%%%%%%

\subsection{Examples of passive linear systems}

\begin{figure}
\centering
\includegraphics[scale=0.37]{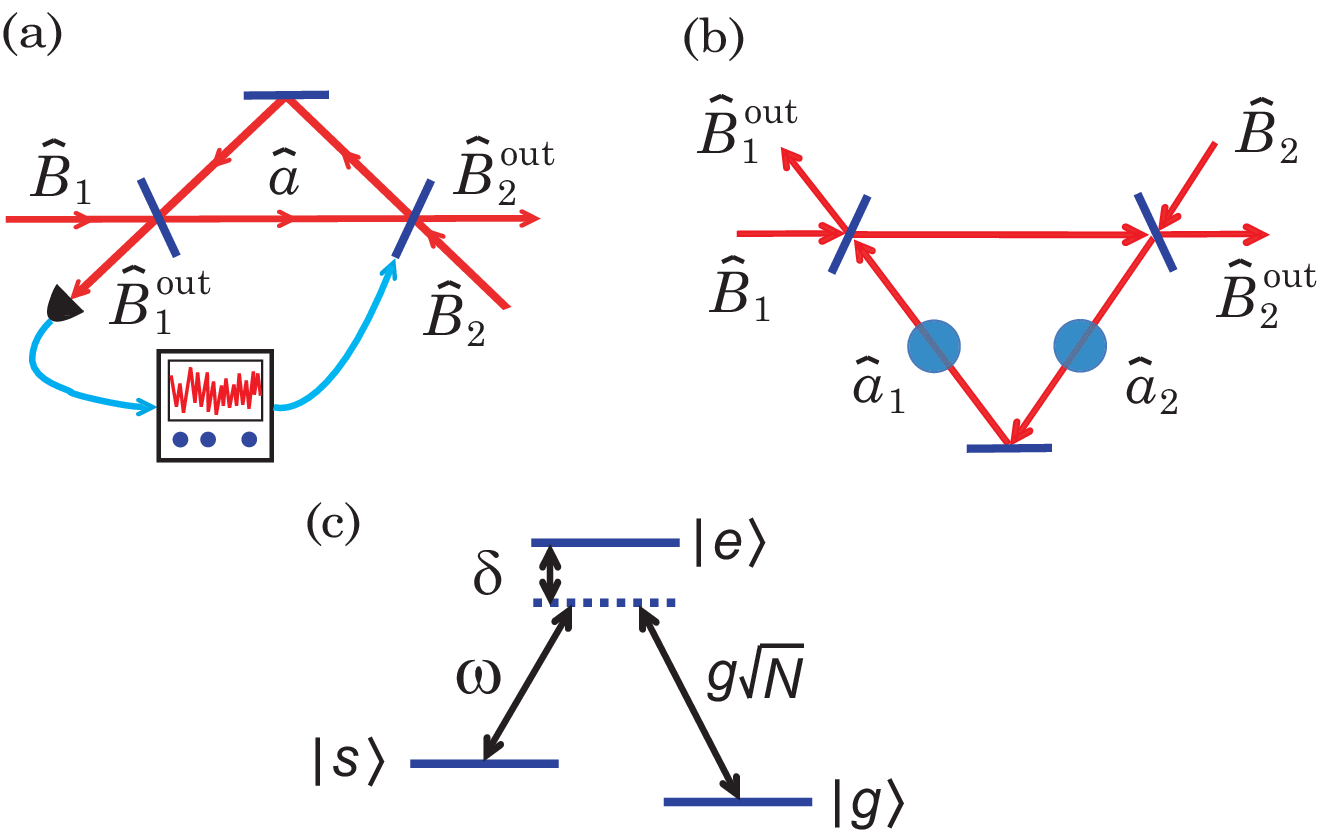}
\caption{
\label{passive examples 1}
Examples of passive linear systems. 
(a) Mode-cleaning cavity; the output field $\hat B_1^{\rm out}$ is measured 
to estimate the detuning $\omega_o$, which is further used to lock the optical 
path length in the cavity. 
(b) Two atomic ensembles; they interact with each other in a nontrivial 
way through the cavity field. 
(c) Energy levels of a $\Lambda$-type atom. 
}
\end{figure}

\begin{example}
\label{cavity example}
The first example is an optical cavity illustrated in 
Fig.~\ref{passive examples 1} (a). 
The intra-cavity field with mode $\hat{a}(t)$ couples to the incoming 
laser field $\hat{B}_1(t)$ and a vacuum $\hat{B}_2(t)$; 
then two outgoing fields $\hat{B}_1^{\rm out}(t)$ and $\hat{B}_2^{\rm out}(t)$ 
appear in the output ports. 
The system dynamics is given by 
\begin{eqnarray}
& & \hspace*{-2em}
     d\hat{a} = (-i\omega_o - \kappa)\hat{a} dt - \sqrt{\kappa}d\hat{B}_1 
                             - \sqrt{\kappa}d\hat{B}_2,
\nonumber \\ & & \hspace*{-2em}
     d\hat{B}_1^{\rm out} = \sqrt{\kappa}\hat{a}dt + d\hat{B}_1,~~
     d\hat{B}_2^{\rm out} = \sqrt{\kappa}\hat{a}dt + d\hat{B}_2, 
\end{eqnarray}
where $\kappa$ is the transmissivity of the coupling mirrors and 
$\omega_o$ is the detuning representing the frequency difference between the 
inner and outer optical fields. 
Note that $C^\dagger=[\sqrt{\kappa}, \sqrt{\kappa}]$ and $\Omega=\omega_o$. 
The role of this cavity system is low-pass filtering for the noisy incoming laser field 
$\hat{B}_1$, and $\hat{B}_2^{\rm out}$ is the resultant mode-cleaned 
field which can be use for quantum information processing \cite{Bachor-2004}. 
To effectively perform mode cleaning, we need to identify the parameter $\omega_o$. 
In practice, the corresponding error signal can be detected by homodyne 
measuring the first output field $\hat{B}_1^{\rm out}$, which is further 
used to lock the cavity path-length to attain $\omega_o=0$ by a piezo-actuator 
mounted on the mirror. Thanks to recent progress in nano-device engineering, it is 
possible to realize high-Q cavities, which can be used for  
storing optical light fields \cite{Noda 2007}. 
\end{example}

\begin{example}
\label{dissipative example}
The next example is that of two large atomic ensembles trapped in a cavity 
(which will be adiabatically eliminated) having 
two input-output ports, as illustrated in Fig.~\ref{passive examples 1}~(b). 
The system variables of the $k$th ensemble $(k=1,2)$ are the total angular 
momentum operators $(\hat{J}_k^x, \hat{J}_k^y, \hat{J}_k^z)$ satisfying 
$[\hat{J}_k^x, \hat{J}_k^y]=i\hat{J}_k^z\sim iJ$~($J\in{\mathbb R}$), 
where the approximation is taken due to the large ensemble limit; 
then, the ``position'' and ``momentum'' operators $\hat q_k=\hat{J}_k^x/\sqrt{J}$, 
$\hat p_k=\hat{J}_k^y/\sqrt{J}$ serve as system variables. 
It was shown in \cite{Parkins,Krauter,Muschik} that a nontrivial coupling 
between the ensembles can be realized, which as a result leads to the following 
dynamical equation: 
\[
      d\hat{\vecx} 
        = -\frac{\kappa}{2}
             \left[ \begin{array}{cc}
                  Y & 0   \\
                  0 & Y   \\
             \end{array} \right]\hat{\vecx}dt
        + i\sqrt{\frac{\kappa}{2}}
             \left[ \begin{array}{cc}
                  -I_2 & I_2   \\
                  iY & iY   \\
             \end{array} \right]
             \left[ \begin{array}{c}
                  d \hat{\bf B}   \\
                  d \hat{\bf B}^*  \\
             \end{array} \right], 
\]
where $\hat{\bf B}=[\hat B_1, \hat B_2]^T$, 
$\hat{\bf B}^*=[\hat B^*_1, \hat B^*_2]^T$, 
\[
      \hat{\vecx}=[\hat q_1, \hat q_2, \hat p_1, \hat p_2]^T,~~
      Y= \left[ \begin{array}{cc}
                  \cosh(2r) & -\sinh(2r)  \\
                  -\sinh(2r) & \cosh(2r)  \\
             \end{array} \right].
\]
and $\kappa$ and $r$ are system parameters. Since $Y>0$, the system is stable and has a unique steady state; 
interestingly, it is the so-called pure {\it two-mode squeezed state} 
\cite{Furusawa-2011}, whose covariance matrix is given by 
$V(\infty) = {\rm diag}\{Y^{-1}/2,~Y/2\}$. 
This implies that the two atomic ensembles are {\it entangled}. 
We emphasize the general fact that, if a linear 
system has a unique pure steady state, then it must be passive 
\cite{Yamamoto2012}. 
Actually, the vector of operators $\hat{\veca}=[\hat a_1, \hat a_2]^T$ 
defined by 
\[
      \hat{\veca}=\frac{1}{\sqrt{2}}[-iY^{1/2},~Y^{-1/2}]\hat{\vecx}
\]
satisfies the CCR  \eqref{CCR} and obeys 
\[
       d\hat{\veca} = -\frac{\kappa}{2}Y\hat{\veca}dt 
           - \sqrt{\kappa}Y^{1/2}d\hat{\bf B},~~
       d\hat{\bf B}^{\rm out}
         = \sqrt{\kappa}Y^{1/2}\hat{\veca} dt + d\hat{\bf B}. 
\]
This is clearly a passive system with $\Omega=0$ and $C=\sqrt{\kappa}Y^{1/2}$. 
(Note that the equation of $\hat{\vecx}$ can be uniquely recovered from that 
of $\hat{\veca}$.) 
Clearly, identifying the parameter $r$ is important, as
it determines the amount of entanglement between the ensembles. 
The same fact holds for the more general case of pure 
{\it Gaussian cluster states}, which may be generated via a passive system 
composed of atomic ensembles \cite{Ficek2009}, can be used for one-way quantum 
computing. 
\end{example}

\begin{example}
\label{memory example}
The last example is that of a medium of 
$N$ $\Lambda$-type atoms trapped in a cavity \cite{Gorshkov}, cf. Fig.~\ref{passive examples 1} (c). 
Each atom has two metastable ground states $\ket{s}$ and $\ket{g}$, and 
an excited state $\ket{e}$. 
The e-g transition is naturally coupled to the cavity mode $\hat a_1$ with 
strength $g\sqrt{N}$, whereas the s-e transition is induced by adding a classical 
magnetic field with time-varying Rabi frequency $\omega(t)$. 
The system's variables are the polarization operator 
$\hat a_2=\hat \sigma_{ge}/\sqrt{N}$ and the spin-wave operator
$\hat a_3=\hat \sigma_{gs}/\sqrt{N}$, where $\hat \sigma_\bullet$ is 
the collective lowering operator. 
As in the previous example, they can be well approximated by annihilation 
operators in the large ensemble limit, and as a result 
$\hat{\veca}=[\hat a_1, \hat a_2, \hat a_3]^T$ obeys the following 
passive system; 
\begin{eqnarray}
& & \hspace*{-2em}
      d\hat{\veca} 
        =   \left[ \begin{array}{ccc}
                  -\kappa & ig\sqrt{N} & 0   \\
                  ig\sqrt{N} & -i\delta & i\omega   \\
                  0 & i\omega^* & 0 \\
             \end{array} \right]\hat{\veca}dt
        - \left[ \begin{array}{c}
                  \sqrt{2\kappa}  \\
                  0 \\
                  0 \\
             \end{array} \right]
                  d \hat B,~~~
\nonumber \\ & & \hspace*{-2em}
      d\hat{B}^{\rm out} = \sqrt{2\kappa}\hat{a}_1dt + d\hat B,
\end{eqnarray}
where $\kappa$ denotes the cavity decay rate and $\delta$ is the detuning 
of the cavity center frequency and the s-e transition frequency. 
This system works as a quantum memory as follows. 
A state of the input optical field $\hat B(t)$ is transferred to that of the 
spin-wave mode $\hat a_3$, and then it is preserved there 
by setting $\omega(t)=0$. 
An effective pulse shaping method for $\omega(t)$ which achieves 
high fidelity state transfer and storage is presented in \cite{Gorshkov}. 
Such an optimal pulse depends on the system's parameters, which 
therefore should be identified as accurately as possible. 
Note that several similar architectures for quantum memory have been 
proposed for instance in an inhomogeneously broadened ensemble of 
atoms or nitrogen-vacancy centers in diamond \cite{Simon,Majer,Hush}, 
nano-mechanical oscillators \cite{He2009}, or a general linear network 
\cite{Yamamoto2014}, all of which are modeled by passive linear systems. 
We should emphasize that the passivity property is essential, as in general 
an active system violates the energy balance and does not realize a perfect 
state transfer. 
\end{example}

%%%%%%%%%%%%%%%%%%%%%%%%%%%%%%%%%%%%%%%%%%%
%%%%%%%%%%%%%%%%%%%%%%%%%%%%%%%%%%%%%%%%%%%
%%%%%%%%%%%%%%%%%%%%%%%%%%%%%%%%%%%%%%%%%%%

\section{The system identifiability}

This section begins with the problem formulation of system identification 
and the definition of identifiability. 
We then provide basic necessary and sufficient conditions for the passive 
linear system \eqref{dynamics} and \eqref{observation} to be identifiable. 
Some examples are given to illustrate the result.

%%%%%%%%%%%%%%%%%%%%%%%%%%%%%%%%%%%%%%%%%%%%

\subsection{System identifiability}

Broadly speaking, by system identification we mean the estimation of the 
parameters $\Omega$ and $C$ which completely characterize the linear 
quantum system \eqref{dynamics} and \eqref{observation}. 
This task can be analyzed in various scenarios, depending on the experimenter's 
ability to prepare the field's input state and the system's initial state, and 
the type of measurements used for extracting information about the dynamics. 
In the simplest experimental scenario the input field is prepared 
in a coherent state with a certain temporal shape
\begin{equation*}
\label{mean input}
    \mean{\hat{\vecb}(t)}=\beta(t),
\end{equation*}
and the experimenter can perform standard (e.g. homodyne and heterodyne) 
measurements on the output. We return to this scenario in section 
\ref{sec.statistics}.

As noted before, in the frequency domain we have 
$\hat{\vecb}^{\rm out}(\omega) = \Xi(-i\omega)\hat{\vecb}(\omega)$, so 
by taking expectation we get 
$\mean{\hat{\vecb}^{\rm out}}(\omega)
= \Xi(-i\omega) \tilde\beta(\omega)$, where $\tilde\beta(\omega)$ is 
the Fourier transform of $\beta(t)$. 
Therefore, the experimenter can at most determine $\Xi(-i\omega)$, and 
this can be done by preparing appropriate inputs (e.g. sinusoids with a 
certain frequency $\omega$), observing the outputs (e.g. by homodyne 
measurements) and computing their Fourier transforms. 
%
%\[
%    {\cal L}[\mean{\hat{\vecb}^{\rm out}}](s)
%     = \Xi(s) {\cal L}[\beta](s). 
%\]
%

In general, the system matrices may be modeled as depending on an 
unknown parameter vector $\theta\in \Theta$ such that 
\begin{equation}
\label{eq.model}
      (\Omega, C)= (\Omega(\theta), C(\theta)), 
\end{equation}
and $\Xi(s)= \Xi(s;\theta)$ correspondingly. 
The task is then to estimate $\theta$ using 
the input and output relations (see Fig.~\ref{Sys ID setup}). 
The identifiability of the system is defined as follows. 
\begin{definition}
\label{def of identifiability}
The parameter $\theta$ is identifiable if 
$\Xi(s;\theta)=\Xi(s;\theta^{\prime})$ for all $s$ implies 
$\theta=\theta^{\prime}$. 
\end{definition}

%%%%%%%%%%%%%%%%%%%%%%%%%%%%%%%%%%%%%%%%%%%%

\subsection{Observability, controllability and minimality}

The concepts of  controllability and observability have a straightforward, 
though arguably non unique, extension to the quantum domain; 
see Section~II-A for the classical case. 
The system defined by \eqref{dynamics} and \eqref{observation} is 
controllable if 
%for any times $t_{0}<t_{1}$ and any means 
%$x_{0},x_{1}\in \mathbb{C}^{n}$ there exists a coherent input which 
%drives the initial coherent state $|x_{0}\rangle$ into the final state 
%$|x_{1}\rangle$ over the time interval $[t_{0},t_{1}]$. 
%This is equivalent to the fact that 
the following \emph{controllability matrix} has full row rank: 
\begin{equation}
\label{controllability matrix}
   {\cal C}=-[C^{\dagger}, AC^{\dagger}, \ldots, A^{n-1} C^{\dagger}].
\end{equation}
Similarly, the system is observable if 
%for any times $t_{0}<t_{1}$, an unknown 
%initial coherent state of the system can be estimated from the history of 
%(coherent) input and output over the time interval $[t_{0},t_{1}]$. 
%This is equivalent to the fact that 
the \emph{observability matrix} 
\begin{equation}
\label{observability matrix}
   {\cal O}=[C^{T}, (CA)^{T}, \ldots, (CA^{n-1})^{T}]^{T}
\end{equation}
has full column rank. 
As in the classical case, if the system is not controllable or observable then 
there exists a lower dimensional system with the same transfer function 
as the original one. 
Thus, we focus on {\it minimal}, i.e. controllable and observable quantum 
systems. 
The following lemma shows that in the passive case we need to check only 
one of the controllability and observability conditions to verify that the 
system is minimal and stable. 
% ; cf. \cite{GutaYamamotoCDC} for the proofs. 
%
\begin{lemma}
\label{lemma 3.1}
For the quantum passive linear system \eqref{dynamics} and 
\eqref{observation}, the controllability and the observability 
conditions are equivalent. 
Moreover, any minimal system is stable, i.e. $A$ is Hurwitz. 
\end{lemma}
\begin{proof}
From the result of systems theory \cite{LjungBook}, $(A,C\dgg)$ controllability 
is equivalent to the following condition: 
$yA=\lambda y,~\exists y, \lambda~\Rightarrow~yC\dgg\neq 0$. 
Then we have 
\begin{equation}
\label{prop1}
    zA\dgg=\mu z,~\exists z, \mu~~\Rightarrow~~zC\dgg\neq 0. 
\end{equation}
To prove \eqref{prop1}, suppose that there exists a vector $z$ satisfying 
$zA\dgg=\mu z$ and $zC\dgg=0$. 
This leads to $z\Omega=-i\mu z$ and $zC\dgg C=0$, yielding 
$zA=z(-i\Omega-C\dgg C/2)=-\mu z$. 
But together with $zC\dgg=0$, this is contradiction to the condition posed 
in the first line, thus \eqref{prop1} holds. 
Now again from the systems theory, \eqref{prop1} is the iff condition for 
$(A\dgg, C\dgg)$ controllability and it is equivalent to $(A,C)$ observability. 
The proof for the inverse direction is the same.

Let us move to prove the stability property. 
Because of the minimality, the system satisfies the condition \eqref{prop1}; 
hence $z^\dagger$ is an eigenvector of $A$ and $\mu^*$ is the 
corresponding eigenvalue. 
Then the relation $zA^\dagger z^\dagger = \mu \|z^\dagger\|^2$ 
together with its complex conjugate lead to 
${\rm Re}(\mu)=-\|Cz^\dagger\|^2/2\|z^\dagger\|^2$, which is 
strictly negative due to $zC^\dagger\neq 0$. 
Therefore $A$ is a Hurwitz matrix. 
\end{proof}

%%%%%%%%%%%%%%%%%%%%%%%%%%%%%%%%%%%%%%%%%%%%

\subsection{The identifiability conditions}

As noted above, by appropriately choosing the input signal $\beta(t)$, 
the observer can effectively identify the transfer function $\Xi(s)$. 
The following theorem gives a precise characterization of systems which 
are equivalent in the sense that they cannot be distinguished based on only 
the input-output relation.

\begin{theorem}
\label{equivalent class} 
Let $(\Omega_1, C_1)$ and $(\Omega_2, C_2)$ be two passive linear 
systems as defined in \eqref{dynamics} and \eqref{observation}, 
and assume that both systems are minimal. 
Then they have the same transfer function if and only if 
there exists a unitary matrix $U$ such that
\begin{equation}
\label{equivalent class transfer}
    \Omega_2=U\Omega_1U\dgg,~~~
    C_2=C_1U\dgg. 
\end{equation}
\end{theorem}

\begin{proof}
It is well known that two minimal systems have the same transfer functions 
\[
    C_1(sI-A_1)^{-1}C_1\dgg=C_2(sI-A_2)^{-1}C_2\dgg, 
\]
(we here omit the trivial constant term $I$) iff there 
exists an invertible matrix $U$ satisfying 
\begin{equation}
\label{similarity trans}
    A_2=UA_1U^{-1},~~C_2\dgg=UC_1\dgg,~~C_2=C_1U^{-1}. 
\end{equation}
Note that $U$ is not assumed to be unitary. 
Using the second and third conditions we have $C_1(U\dgg U)=C_1$, 
which further gives $[U\dgg U,~C_1\dgg C_1]=0$. 
Also, applying the second and third conditions to the first one, we 
have $\Omega_2=U\Omega_1 U^{-1}$. 
Then, because $\Omega_i$ is a Hermitian matrix, $[U\dgg U,~\Omega_1]=0$ 
holds. 
Combining these two results we obtain $[U\dgg U,~A_1]=0$. Therefore we have 
\[
    C_1 A_1=C_1(U\dgg U)A_1=C_1A_1(U\dgg U), 
\]
which means that the observability matrix ${\cal O}$ satisfies 
${\cal O}={\cal O}U\dgg U$. 
Because of the assumption that ${\cal O}$ is of full rank, $U$ is unitary. 
Therefore the conditions \eqref{similarity trans} are reduced to 
\eqref{equivalent class transfer}. 
\end{proof}

For a parameterized model the identifiability condition is given by the following.

\begin{corollary}
\label{similarity transformation}
Let $(\Omega(\theta), C(\theta))$ be a minimal system with unknown 
parameter vector $\theta\in \Theta$. 
Then $\theta$ is identifiable if and only if 
\begin{equation*}
    \Omega(\theta^{\prime})=U\Omega(\theta)U\dgg,~~~
    C(\theta^{\prime})=C(\theta)U^\dagger
\end{equation*}
implies $\theta=\theta^{\prime}$. 
\end{corollary}

The above result can be interpreted as follows. 
The matrix $U$ corresponds to the coordinate transformation 
$\hat{\veca}'=U\hat{\veca}$ and the unitarity of $U$ means that the 
canonical commutation relation \eqref{CCR} is preserved. 
Note that if the system variables contain classical components, $U$ would not 
necessarily be unitary. 
Similarly, if the system is not passive, then one needs to consider both 
$\hat{\veca}$ and $\hat{\veca}^{*}$ as coordinates, and corresponding 
doubled-up transfer matrices \cite{GoughPRA2010}.

In addition to the above corollary, we give another criterion for 
testing the identifiability. 
%, cf  \cite{GutaYamamotoCDC} for the proof. 
Note this result does not require the minimality of the system.

\begin{lemma}
\label{simple test}
The parameter $\theta$ is identifiable if and only if 
\begin{equation}
\label{markov parameters}
    C(\theta)\Omega(\theta)^kC(\theta)\dgg
     =C(\theta^{\prime})\Omega(\theta^{\prime})^kC(\theta^{\prime})\dgg,
     ~~~\forall k
\end{equation}
implies $\theta=\theta^{\prime}$. 
\end{lemma}

\begin{proof}
For simplicity let us denote $C:= C(\theta), C^{\prime}:= C(\theta^{\prime})$ 
and similarly for $\Omega$ and  $A$.
By expanding the equation $\Xi(s;\theta)=\Xi(s;\theta^{\prime})$ 
with respect to $s$ and comparing their coefficients, we have 
$CA^k C \dgg=C^{\prime}A^{\prime k} C^{\prime \dagger}$ for all $ k$, 
and thus 
\[
   C\big(-i\Omega-\half C\dgg C\big)^k C\dgg=
      C^{\prime}\big(-i\Omega^{\prime} 
          - \half C^{\prime\dagger} C^{\prime}\big)^k C^{\prime\dagger}. 
\]
This $k$-th order polynomial is composed of the linear combination of 
$C[(C\dgg C)^p\circ\Omega^q]C\dgg$ with $p+q=k$, where $\circ$ 
means the symmetrization, e.g. 
$(C\dgg C)^1\circ\Omega^2 = (C\dgg C)\Omega^2 
+\Omega(C\dgg C)\Omega +\Omega^2(C\dgg C)^2$ for $k=3$. 
Then \eqref{markov parameters} can be proven by induction with 
respect to $k$. 
\end{proof}

%%%%%%%%%%%%%%%%%%%%%%%%%%%%%%%%%%%%%%%%%%%%%

\subsection{Examples}

We here apply the identifiability conditions to some systems. 
The critical assumption is that we have some a priori information about 
the system, such as the structure of the network and some parameters. 
This a priori knowledge helps us to reduce the size of the equivalence 
class of the system and in some cases even to exactly identify the 
system, as will be demonstrated.

\begin{example}
We begin with the simple cavity system studied in Example \ref{cavity example}. 
In this case, $\Omega=\omega_o$ and 
$C^\dagger=[\sqrt{\kappa}, \sqrt{\kappa}]^T$, where we assume that 
$\kappa$ is a known parameter. 
Now, from Theorem~\ref{equivalent class}, the equivalence class is generated 
by a trivial $1\times 1$ unitary matrix $U=e^{i\phi}$; 
but clearly $C=CU^\dagger$ imposes $U=1$, hence 
from Corollary~\ref{similarity transformation} $\omega_o$ is identifiable. 
\end{example}

\begin{example}
Next let us consider the system in Example~\ref{dissipative example}, 
where $\Omega=0$ and $C=\sqrt{\kappa}Y^{1/2}$. 
It is easy to see that the system is minimal. 
Then Theorem~\ref{equivalent class} states that the equivalence class is 
generated by a unitary matrix $U$ as 
\[
     \Omega'=0,~~~
     C'=CU^\dagger
        =\sqrt{\kappa}
          \left[ \begin{array}{cc}
                  \cosh(r) & -\sinh(r)  \\
                  -\sinh(r) & \cosh(r)  \\
             \end{array} \right] U^\dagger. 
\]
Now, we know that $C'$ is positive symmetric and the (1,1) and (2,2) elements 
are the same; 
this a priori knowledge allows only $U=I_2$, so the parameters are identifiable. 
\end{example}

\begin{example}
The memory system shown in Example~\ref{memory example} is a passive 
system essentially with 
\begin{equation}
\label{C matrix for 3 nodes}
    C=[\sqrt{2\kappa},~0,~ 0],~~
    \Omega(\theta)
    =\left[ \begin{array}{ccc}
       0 & \theta_1 & 0        \\
       \theta_1 & 0 & \theta_2 \\
       0 & \theta_2 & 0        \\
     \end{array} \right], 
\end{equation}
where $(\theta_{1}, \theta_{2})$ are unknown coupling constants 
to be identified (we assume $\delta=0$).

We immediately see that the system is controllable and accordingly minimal. 
Thus, we can apply Theorem~\ref{equivalent class}, showing that the 
equivalence class of the system is generated by the unitary matrix $U$. 
But since we know the structure of the matrices $\Omega$ and $C$, 
it follows that $U$ must be either $U_1= {\rm Diag}(1,1,1)$, 
$U_2= {\rm Diag}(1,-1,1)$, $U_3={\rm Diag}(1,1,-1)$, or 
$U_4= {\rm Diag}(1,-1,-1)$. 
This means that the systems with parameter $\theta=(\theta_1,\theta_2)$, 
$(-\theta_1, \theta_2), (\theta_1,-\theta_2)$, and 
$(-\theta_1,-\theta_2)$ have the same transfer function. 
Therefore the parameters $\theta_1$ and $\theta_2$ are identifiable 
up to the sign, i.e. $\theta$ is locally identifiable but {\it not globally} 
\cite{Glover1974}.

An alternative proof of the above result is obtained by using 
Lemma~\ref{simple test}. 
Actually we compute 
%
%\[
%    C\Omega(\theta) C\dgg = 0,~~~
%    C\Omega(\theta)^{2} C\dgg = 2\kappa\theta_1^2,~~~
%    C\Omega(\theta)^3 C\dgg = 0,~~~
%    C\Omega(\theta)^4 C\dgg = 2\kappa\theta_1^2(\theta_1^2+\theta_2^2),
%\]
%
%% For double-column format:
%
\begin{eqnarray}
& & \hspace*{-1em}
    C\Omega(\theta) C\dgg = 0,~~~
    C\Omega(\theta)^{2} C\dgg = 2\kappa\theta_1^2,
\nonumber \\ & & \hspace*{-1em}
    C\Omega(\theta)^3 C\dgg = 0,~~~
    C\Omega(\theta)^4 C\dgg = 2\kappa\theta_1^2(\theta_1^2+\theta_2^2)
\nonumber
\end{eqnarray}
yielding $\theta_1^{2}=\theta_1^{\prime 2}$ and 
$\theta_2^{2}=\theta_2^{\prime 2}$ hold, if $\theta_1\neq 0$. 
Thus we have the same conclusion as above.

A third route is to look directly at the transfer function: 
\[
    \Xi(s)
     =\frac{s^3-\kappa s^2+(\theta_1^2+\theta_2^2)s-\kappa\theta_2^2}
           {s^3+\kappa s^2+(\theta_1^2+\theta_2^2)s+\kappa\theta_2^2},
\]
and note that the poles give us enough information to determine both 
$\theta_1^{2}$ and $\theta_2^{2}$. 
Note when $\theta_1=0$ (i.e., there is no connection between $\hat a_1$ 
and $\hat a_2$), $\Xi(s)=(s-\kappa)/(s+\kappa)$, showing that the system 
is clearly not minimal; 
actually in this case $\theta_2$ cannot be estimated. 
\end{example}

\begin{figure}
\centering
\includegraphics[scale=0.35]{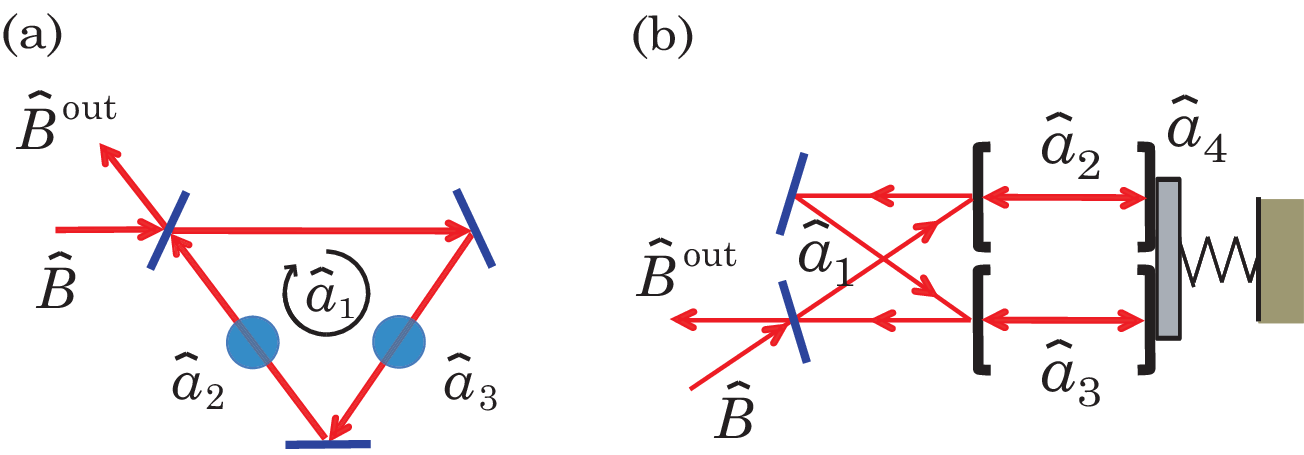}
\caption{\label{Passive Examples 2}
Examples of passive linear systems. 
(a) Two atomic ensembles where in this case the cavity field with mode 
$\hat a_1$ is not adiabatically eliminated. 
(b) Opto-mechanical oscillator with phonon mode $\hat a_4$, which is 
coupled to two cavities with modes $(\hat a_2, \hat a_3)$; they are 
further coupled to a bow-tie type cavity with mode $\hat a_1$, which 
works as an input-output port. 
}
\end{figure}

\begin{example}
Let us consider the large atomic ensemble network depicted in 
Fig.~\ref{Passive Examples 2}~(a). 
The cavity field $\hat a_1$ is coupled to the input field and is connected 
to the ensembles with modes $\hat a_2$ and $\hat a_3$ which correspond 
to the collective lowering operators of the ensembles \cite{Parkins}. 
The system Hamiltonian is given by 
$\hat{H} = \Delta \hat{a}_2^*\hat{a}_2
             +\theta_1(\hat{a}_1^*\hat{a}_2+\hat{a}_1\hat{a}_2^*)
             +\theta_2(\hat{a}_1^*\hat{a}_3+\hat{a}_1\hat{a}_3^*)$, 
hence we have 
\begin{equation*}
%\label{example 2 matrix Omega}
    \Omega(\theta)
    =\left[ \begin{array}{ccc}
       0 & \theta_1 & \theta_2 \\
       \theta_1 & \Delta & 0        \\
       \theta_2 & 0 & 0        \\
     \end{array} \right]. 
\end{equation*}
The $C$ matrix is the same as in \eqref{C matrix for 3 nodes}.

The additional detuning Hamiltonian $\Delta \hat{a}_2^*\hat{a}_2$ is 
necessary for the parameters $\theta_1$ and $\theta_2$ to be 
identifiable, because the system is minimal only when $\Delta\neq 0$. 
In fact, when $\Delta=0$ we cannot distinguish the two ensembles, 
thus the system is not identifiable. 
So we assume $\Delta\neq 0$ and apply Theorem~\ref{equivalent 
class}. 
The constraint $C=CU^\dagger$ implies that $U$ must be of the form 
$U={\rm Diag}(1, \tilde{U})$ with $\tilde{U}$ a $2\times 2$ unitary matrix. 
Then the equivalence class is characterized by 
\[
    \Omega'
    =\left[ \begin{array}{cc}
       1 & 0^T \\
       0 & \tilde{U}   \\
     \end{array} \right]
     \left[ \begin{array}{cc}
       0 & \theta^T \\
       \theta & \Lambda \\
     \end{array} \right]
     \left[ \begin{array}{cc}
       1 & 0^T \\
       0 & \tilde{U}^\dagger \\
     \end{array} \right]
    =\left[ \begin{array}{cc}
       0 & (\tilde{U}\theta)^\dagger \\
       \tilde{U}\theta & \tilde{U}\Lambda \tilde{U}^\dagger \\
     \end{array} \right], 
\]
where $\theta=[\theta_1, \theta_2]^T$ and 
$\Lambda={\rm Diag}(\Delta, 0)$. 
Now we know that the matrix $\Omega'$ is of the same 
form as $\Omega$, which yields additional constraint on $\tilde{U}$, 
i.e. $\tilde{U}\Lambda \tilde{U}^\dagger=\Lambda$, or equivalently 
$[\tilde{U}, \Lambda]=0$. 
This readily clarifies that $\tilde{U}$ is diagonal; 
hence together with $\tilde{U}\theta\in{\mathbb R}^2$, we conclude that 
the parameters $\theta_{1}$ and $\theta_{2}$ are identifiable up to 
the sign.
\end{example}

\begin{example}
The last example is a linear network composed of cavities and an 
opto-mechanical oscillator shown in Fig.~\ref{Passive Examples 2}~(b). 
This specific configuration is inspired by \cite{Hamerly2013} and the 
oscillator can serve as a quantum memory. 
The oscillator with phonon mode $\hat a_4$ couples to two cavities 
with modes $(\hat a_2, \hat a_3)$, through radiation pressure force; 
particularly with the dissipative (red-sideband) regime the coupling 
Hamiltonian takes a passive form \cite{MilburnWoolley}. 
The two cavities further interact with a bow-tie type cavity with mode 
$\hat a_1$. 
As a result, the system Hamiltonian is given by 
%
%\[
%      \hat{H} = \theta_1(\hat{a}_1^*\hat{a}_2+\hat{a}_1\hat{a}_2^*)
%              +\theta_2(\hat{a}_1^*\hat{a}_3+\hat{a}_1\hat{a}_3^*)
%              +\theta_3(\hat{a}_2^*\hat{a}_4+\hat{a}_2\hat{a}_4^*)
%              +\theta_4(\hat{a}_3^*\hat{a}_4+\hat{a}_3\hat{a}_4^*), 
%\]
%
% For double column format:
%
\begin{eqnarray}
& & \hspace*{-1em}
      \hat{H} = \theta_1(\hat{a}_1^*\hat{a}_2+\hat{a}_1\hat{a}_2^*)
              +\theta_2(\hat{a}_1^*\hat{a}_3+\hat{a}_1\hat{a}_3^*)
\nonumber \\ & & \hspace*{1em}
      \mbox{} +\theta_3(\hat{a}_2^*\hat{a}_4+\hat{a}_2\hat{a}_4^*)
              +\theta_4(\hat{a}_3^*\hat{a}_4+\hat{a}_3\hat{a}_4^*), 
\nonumber
\end{eqnarray}
thus we have 
\begin{equation*}
\label{example 2 matrix Omega}
    \Omega(\theta)
    =\left[ \begin{array}{cccc}
       0            & \theta_1 & \theta_2 & 0        \\
       \theta_1     & 0        & 0        & \theta_3 \\
       \theta_2     & 0        & 0        & \theta_4 \\
       0            & \theta_3 & \theta_4 & 0        \\
     \end{array} \right], 
\end{equation*}
while the $C$ matrix is given by $C=[\sqrt{2\kappa},~0,~ 0, ~0]$.

Let us first check the minimality. 
A direct computation shows that the observability matrix 
${\cal O}$ satisfies 
${\rm det}({\cal O})
=4\kappa^2(\theta_1\theta_3+\theta_2\theta_4)^2
          (\theta_2\theta_3-\theta_1\theta_4)$. 
Hence, we consider the minimal system satisfying 
${\rm det}({\cal O})\neq 0$. 
Then from Theorem~\ref{equivalent class}, the equivalence class 
is generated in terms of the unitary $U={\rm Diag}(1, \tilde{U})$ with 
$\tilde{U}$ a $3\times 3$ unitary matrix, and it is parameterized by 
\[
    \Omega'
    =\left[ \begin{array}{cc}
       0 & [\theta_{12}^T~~0]\tilde{U}^\dagger \\
       \tilde{U} \left[ \begin{array}{c}
          \theta_{12} \\
          0 \\
         \end{array} \right]
        & \tilde{U}\Theta \tilde{U}^\dagger \\
     \end{array} \right],~~
    \Theta=\left[ \begin{array}{cc}
             0 & \theta_{34} \\
             \theta_{34}^T & 0 \\
     \end{array} \right], 
\]
where $\theta_{12}=[\theta_1, \theta_2]^T$, 
$\theta_{34}=[\theta_3, \theta_4]^T$. 
The structure of the matrix $\Omega'$ further imposes the additional 
constraint on $\tilde{U}$, which as a result yields $\tilde{U}={\rm Diag}(V, 1)$ 
with $V$ a $2\times 2$ orthogonal matrix. 
Therefore, the equivalence class is the system whose Hamiltonian 
matrix is characterized by 
\[
   \Omega'
    =\left[ \begin{array}{ccc}
       0 & \theta_{12}^T V^T & 0 \\
       V\theta_{12} & O & V\theta_{34} \\
       0 & \theta_{34}^T V^T & 0 \\
     \end{array} \right]. 
\]
Hence, from Theorem~\ref{equivalent class}, the systems specified by 
$(\Omega', C)$ have the same transfer function for all $V$. 
Thus, this system is not (completely) identifiable. 
However, if for instance the second cavity mode $\hat a_2$ is detuned and 
as consequence the (2,2) element of $\Omega$ is nonzero, then the system 
gains the identifiability property. 
\end{example}

%%%%%%%%%%%%%%%%%%%%%%%%%%%%%%%%%%%%%%%%%%%
%%%%%%%%%%%%%%%%%%%%%%%%%%%%%%%%%%%%%%%%%%%
%%%%%%%%%%%%%%%%%%%%%%%%%%%%%%%%%%%%%%%%%%%

\section{Network identification; the infection condition}

As demonstrated in Section III, in order to establish the identifiability 
of a given system, we need to carry out certain model specific calculations 
ruling out the existence of non-trivial unitaries in Theorem 3.1. 
It would therefore be useful to find an identifiability criterion which applies 
to a general class of systems. 
In this section we describe such a criterion which relies on the special 
topological structure of the Hamiltonian. 
Similar results have been found in different contexts 
\cite{Burgarth2009NJP,Burgarth2011NJP}.

Let $\mathcal{V}$ be the set of vertices representing the modes 
of our continuous variables system. 
The interactions between the different modes are modeled by 
the set of edges $\mathcal{E}$ over $\mathcal{V}$: 
$\mathcal{E} \subset \mathcal{V}\times \mathcal{V}$, so that two 
modes $i$ and $j$ interact if they are connected by an edge. 
More precisely, we assume that the matrix $\Omega$ describing 
the system Hamiltonian is of the form
\begin{equation}
\label{eq.omega.graph}
  \Omega(\theta)= 
     \sum_{(i,j)\in \mathcal{E}} \omega_{i,j}(\theta) 
        ( e_{i} e_{j}^T + e_{j} e_{i}^T ),
\end{equation}
where $\omega_{i,j}(\theta)$ are unknown \emph{real} coefficients 
which make up the parameter $\theta$ and $e_i=[0, \cdots, 1, \cdots, 0]^T$ 
is the basis vector having zeros except the $i$th element. 
We further assume that the coupling between the system and the 
field is known and specified by the matrix $C$ whose support is 
spanned by a set of basis vectors $\{ e_{i} : i \in\mathcal{I}\}$ 
for some set of vertices $\mathcal{I}$, the restriction of $C^\dagger C$ to 
this subspace being strictly positive.

\begin{figure}
\centering
\includegraphics[scale=0.3]{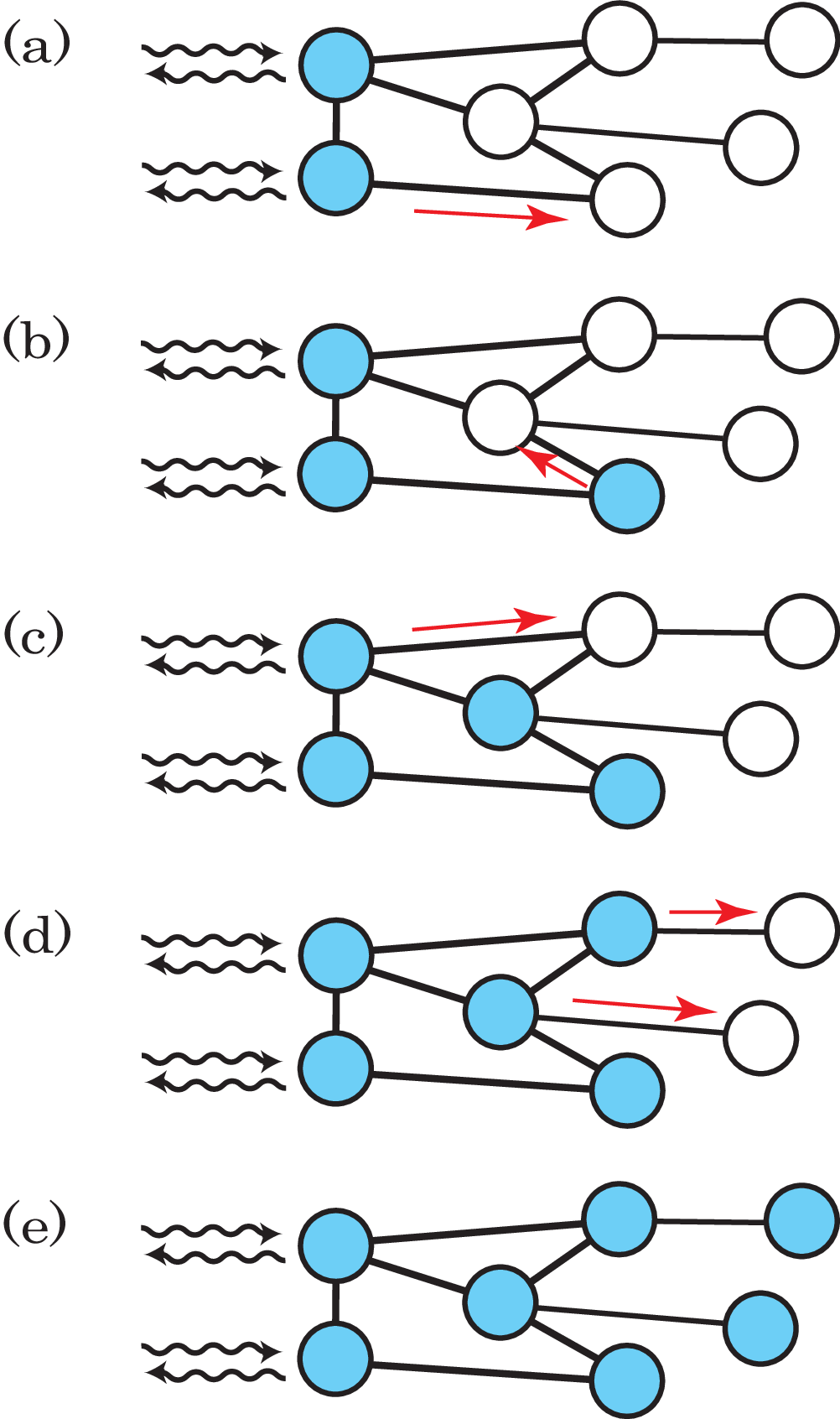}
\caption{
\label{Infection graph}
Infection property. 
The colored node indicates that it is infected, and the arrow indicates 
that the infection occurs along that edge. 
Through the steps from (a) to (e), the whole network becomes infected. 
}
\end{figure}

The crucial property we will require of $\mathcal{I}$ is that 
it is {\it infecting} for the graph $(\mathcal{V}, \mathcal{E})$, 
which can be defined sequentially by the following conditions 
(see Fig.~\ref{Infection graph}): 
\begin{itemize}
\item[(i)] At the beginning the vertices in $\mathcal{I}$ are infected; 

\item[(ii)] If an infected vertex has only one non-infected neighbor, 
the neighbor gets infected;

\item[(iii)] After some interactions all nodes end up infected. 
\end{itemize}
Roughly speaking, this infection property means that the network 
is similar to a ``chain", where the neighboring nodes are coupled 
to each other. 
Such a chain structure often appears in practical situations, and 
as shown in \cite{BurgarthPRA2009b}, it can be fully controlled by 
only accessing to its local subsystem. 
Also it is notable that in general a chain structure realizes 
fast spread of quantum information \cite{Murphy} and is thus 
suitable for e.g. distributing quantum entanglement. 
The result we present here is that such a useful network is 
always identifiable.

\begin{lemma}
\label{infection lemma}
Let $\Omega(\theta)$ be given by \eqref{eq.omega.graph}, and 
assume that the support of $C$ is spanned by
$\{ e_{i} :i \in\mathcal{I}\}$ with 
$(\mathcal{I},\mathcal{V},\mathcal{E})$ having the infecting 
property. 
Then, the system is minimal. 
\end{lemma}

\begin{proof}
From the assumption, at least one vertex $i_{0}\in\mathcal{I}$ is 
connected to exactly one vertex $j_{0}\in \mathcal{I}^{c}$. 
Thus, $\Omega(\theta)$ can be written as
%
%\begin{eqnarray*}
%   \Omega(\theta)
%     &=& 
%       \omega_{i_{0},j_{0}}(\theta)
%          ( e_{i_{0}} e_{j_{0}}^T + e_{j_{0}} e_{i_{0}}^T)
%      +  \sum_{i\in \mathcal{I},i\neq i_{0}}\sum_{j\in \mathcal{I}^{c}} 
%          \omega_{i,j}(\theta) (e_{i} e_{j}^T +e_{j} e_{i}^T)
%   \\ &+&
%       \sum_{i, j\in \mathcal{I}} 
%          \omega_{i,j}(\theta) (e_{i} e_{j}^T +e_{j} e_{i}^T)
%      + \sum_{i,j\in \mathcal{I}^{c}} 
%          \omega_{i,j}(\theta) (e_{i} e_{j}^T +e_{j} e_{i}^T). 
%\end{eqnarray*}
%
% For double column format:
%
\begin{eqnarray*}
   \Omega(\theta)
     &=& 
       \omega_{i_{0},j_{0}}(\theta)
          ( e_{i_{0}} e_{j_{0}}^T + e_{j_{0}} e_{i_{0}}^T)
   \\ &+&
       \sum_{i\in \mathcal{I},i\neq i_{0}}\sum_{j\in \mathcal{I}^{c}} 
          \omega_{i,j}(\theta) (e_{i} e_{j}^T +e_{j} e_{i}^T)
   \\ &+&
       \sum_{i, j\in \mathcal{I}} 
          \omega_{i,j}(\theta) (e_{i} e_{j}^T +e_{j} e_{i}^T)
   \\ &+&
       \sum_{i,j\in \mathcal{I}^{c}} 
          \omega_{i,j}(\theta) (e_{i} e_{j}^T +e_{j} e_{i}^T). 
\end{eqnarray*}
This readily leads to 
\[
   \Omega(\theta) e_{i_0}=\omega_{i_0,j_0}(\theta) e_{j_0}
    + 2\sum_{j\in \mathcal{I}} 
          \omega_{i_0,j}(\theta) e_{j}.
\]
Also clearly $C^\dagger C e_{i_0}$ is spanned by the vectors $\{e_i : i\in{\cal I}\}$. 
These two facts imply that $A e_{i_0} = (-i\Omega - C^\dagger C/2) e_{i_0}$ is 
spanned by $e_{j_0}$ and $\{e_i : i\in{\cal I}\}$. 
In other words, $C^\dagger$ and $Ae_{i_0}$ generate a new infecting 
set ${\cal I}'={\cal I}\cup\{j_0\}$. 
Repeating this procedure, we find that the controllability matrix 
\eqref{controllability matrix}, 
${\cal C}=-[C^\dagger, AC^\dagger, \ldots, A^{n-1}C^\dagger]$, is 
of full rank, thus the system is controllable. 
This further implies from Lemma~\ref{lemma 3.1} that the system 
is observable, thus as a result it is minimal. 
\end{proof}

\begin{theorem}
\label{infection theorem}
Let $\Omega(\theta)$ be given by \eqref{eq.omega.graph}, and 
assume that the support of $C$ is spanned by
$\{ e_{i} :i \in\mathcal{I}\}$ with 
$(\mathcal{I},\mathcal{V},\mathcal{E})$ having the infecting 
property. 
Then, $\Omega(\theta)$ is identifiable.
\end{theorem}

\begin{proof}
First, from Lemma~\ref{infection lemma} we can apply 
Theorem \ref{similarity transformation}; 
the two parameters are in the same equivalence class if and only if 
there exists an $n\times n$ unitary matrix $U$ such that 
\begin{equation}
\label{eq.transf.hamiltonian}
    \Omega(\theta_{2})= U\Omega(\theta_{1})U^{\dagger},
\end{equation}
and $C= CU$. 
The latter condition implies $[U, C^\dagger C]= 0$ and in particular $U$ 
commutes with projection $P$ onto the support of $C^\dagger C$ so that
\begin{equation}\label{eq.u}
   U=\left[ \begin{array}{c|c}
           I & 0 \\ \hline
           0 & V \\
            \end{array} \right]
\end{equation}
with $V$ unitary on the orthogonal complement of the support of $C$.  
Let us write the Hamiltonian in the block form according to the 
partition $\mathcal{J}=\mathcal{I}\cup \mathcal{I}^{c}$: 
\begin{equation*}
    \Omega(\theta) 
        = \left[ \begin{array}{c|c}
              \Omega_{11}(\theta) & \Omega_{12}(\theta) \\ \hline
              \Omega_{21} (\theta)&  \Omega_{22}(\theta) \\
          \end{array} \right].
\end{equation*}
Then \eqref{eq.transf.hamiltonian} implies that 
%
%\begin{equation}
%\label{Omega_ij}
%    \Omega_{11}(\theta_{2}) = \Omega_{11}(\theta_{1}),~~~
%    \Omega_{12}(\theta_{2}) = \Omega_{12}(\theta_{1})V^{\dagger},~~~
%    \Omega_{22}(\theta_{2}) = V \Omega_{22}(\theta_{1}) V^{\dagger}. 
%\end{equation}
%
% For double column format
%
\begin{eqnarray}
\label{Omega_ij}
    \Omega_{11}(\theta_{2})&=&  \Omega_{11}(\theta_{1}), 
\nonumber \\
    \Omega_{12}(\theta_{2})&=& \Omega_{12}(\theta_{1})V^{\dagger}, 
\nonumber \\
    \Omega_{22}(\theta_{2})&=& V   \Omega_{22}(\theta_{1}) V^{\dagger}. 
\end{eqnarray}
The first equation of \eqref{Omega_ij} means that 
\begin{equation}
\label{eq.ii}
    \omega_{i,j}(\theta_{1}) 
       = \omega_{i,j}(\theta_{2}), \qquad i,j\in\mathcal{I}.
\end{equation}
Furthermore, since $\mathcal{I}$ is infecting, there exists at least one 
vertex $i_{0}\in\mathcal{I}$ which is connected to exactly one vertex 
$j_{0}\in \mathcal{I}^{c}$, so that  the off-diagonal block 
$\Omega_{12}(\theta)$ can be written as
%
%\[   
%    \left( \begin{array}{c|c}
%           0 & \Omega_{12}(\theta) \\ \hline
%             0&  0 \\
%          \end{array} \right)
%     = \omega_{i_{0},j_{0}}(\theta)
%         ( e_{i_{0}} e_{j_{0}}^T + e_{j_{0}} e_{i_{0}}^T)
%     + \sum_{i\in \mathcal{I},i\neq i_{0}}
%        \sum_{j\in \mathcal{I}^{c}} \omega_{i,j}(\theta) (e_{i} e_{j}^T +e_{j} e_{i}^T). 
%\]
%
% For double column format:
%
\begin{eqnarray*}
   \left[ \begin{array}{c|c}
           0 & \Omega_{12}(\theta) \\ \hline
             0&  0 \\
          \end{array} \right]
     &=& 
       \omega_{i_{0},j_{0}}(\theta)
         ( e_{i_{0}} e_{j_{0}}^T + e_{j_{0}} e_{i_{0}}^T)
   \\&+&
      \sum_{i\in \mathcal{I},i\neq i_{0}}
      \sum_{j\in \mathcal{I}^{c}} \omega_{i,j}(\theta) (e_{i} e_{j}^T +e_{j} e_{i}^T). 
\end{eqnarray*}
The second equation of \eqref{Omega_ij} then implies
\[
      \omega_{i_{0},j_{0}}(\theta_{1}) U e_{j_{0}} 
         = \omega_{i_{0},j_{0}}(\theta_{2}) e_{j_{0}}, 
\]
which means that $e_{j_{0}}$ is an eigenvector of $U$ and 
$\omega_{i_{0},j_{0}}(\theta_{2}) 
= \exp(i\phi_{0})\omega_{i_{0},j_{0}}(\theta_{1})$ 
for some phase $\phi_{0}$. 
But since the coefficients of $\Omega(\theta)$ are assumed to be real, 
this implies that 
\begin{equation}
\label{eq.ij}
   \omega_{i_{0},j_{0}}(\theta_{1})=\omega_{i_{0},j_{0}}(\theta_{2}),
      \qquad i_0\in\mathcal{I},~~j_0\in\mathcal{I}^c.
\end{equation}
Additionally, since $U e_{j_{0}} = e_{j_{0}}$, a decomposition of the form 
\eqref{eq.u} holds with the identity block supported by the index set 
$\mathcal{I}^{\prime} = \mathcal{I}\cup\{j_{0}\}$.

The same argument can now be repeated for the set $\mathcal{I}^{\prime}$, 
and by using the infecting property, all vertices will be eventually included 
in the growing set of indices, so that at the end we have 
$\Omega(\theta_{1}) = \Omega(\theta_{2})$. 
%the identity 
%\eqref{eq.22} implies that
%%
%\begin{equation}
%\label{eq.jj}
%     \omega_{j_{i},j_{k}}(\theta_{1})
%        =\omega_{j_{i},j_{k}}(\theta_{2}), \qquad i,k\in \mathcal{I}.
%\end{equation}
%%
%From Eqs.~\eqref{eq.ii}, \eqref{eq.ij} and \eqref{eq.jj}, we conclude 
%that the two Hamiltonians coincide on the subset of vertices 
%$\mathcal{I}\cup \mathcal{I}_{\partial}$, where 
%%
%$$
%    \mathcal{I}_{\partial}= \{j_{i} :i \in\mathcal{I}\},
%$$
%%
%is the set of neighboring vertices of $\mathcal{I}$. 
%
%
%
%Since $\mathcal{I}\cup \mathcal{I}_{\partial}$ has the infection 
%property as well, we can repeat the same argument to enlarge the 
%set of vertices on which the Hamiltonian coincide and after a 
%finite number of interactions, we conclude that 
%$\Omega(\theta_{1}) = \Omega(\theta_{2})$. 
Consequently, from Corollary \ref{similarity transformation}, 
the system is identifiable. 
\end{proof}

From this result, we now readily see that the system in Example~3.3 in 
Section~III-D is identifiable, since clearly this system has a chain-type 
structure and is thus infecting. 
On the other hand, the systems of Examples 3.4 and 3.5 have the tree and 
ring structures, respectively, which are thus not infecting. 
Hence, Theorem~\ref{infection theorem} states nothing about the 
identifiability of these systems; 
in fact, as shown there, the tree system is identifiable, while the ring 
one is not.

%%%%%%%%%%%%%%%%%%%%%%%%%%%%%%%%%%%%%%%%%%%
%%%%%%%%%%%%%%%%%%%%%%%%%%%%%%%%%%%%%%%%%%%
%%%%%%%%%%%%%%%%%%%%%%%%%%%%%%%%%%%%%%%%%%%

\section{Methods for system matrices identification}

Let us consider the situation where we have constructed the transfer function 
matrix $\Xi(s)$, using the input-output data; this is indeed possible via 
several techniques \cite{LjungBook}. 
In the SISO case, this means that we have determined the coefficients 
$(a_i, c_i)$ of the following rational function: 
\[
   \Xi(s)=1+\frac{c_{n-1}s^{n-1} + \cdots + c_0}
             {s^n + a_{n-1}s^{n-1} + \cdots + a_0}. 
\]
Then the following set of system matrices 
\begin{eqnarray}
\label{typical choice}
& & \hspace*{-1em}
       A_0=\left[ \begin{array}{cccc}
       0 & 1 & & 0 \\
         & \ddots & \ddots & \\
       0 & & 0 & 1 \\
       -a_0 & -a_1 &  & -a_{n-1} \\
     \end{array} \right],~~
   B_0=\left[ \begin{array}{c}
       0 \\
       \vdots \\
       0 \\
       1 \\
     \end{array} \right],~~
\nonumber \\[2mm] & & \hspace*{-1em}
   C_0 = [ c_0, \cdots, c_{n-1} ], 
\end{eqnarray}
constitute a realization of $\Xi(s)$ in the sense that 
$\Xi(s)=1+C_0(sI-A_0)^{-1}B_0$. 
Any other realization having the same transfer function can be generated 
via the similarity transformation
\begin{equation}
\label{typical trans}
   A=TA_0T^{-1},~~
   B=TB_0,~~
   C=C_0T^{-1}. 
\end{equation}
However, the matrices \eqref{typical choice} do not satisfy the constraints 
imposed on passive linear quantum systems. 
This means that, for general $T$, the transformation \eqref{typical trans} 
does not yield the set of coefficient matrices of a quantum system; 
e.g., the relation $B=-C^\dagger$ is not satisfied. 
Clearly, in this case, the system matrices $(\Omega, C)$ cannot be 
reconstructed. 
This is an important issue, since from the physics viewpoint we are often 
interested in the system matrices and the system parameters, rather 
than the transfer function. 
Therefore, we need to find a special class of $T$ so that the coefficient 
matrices \eqref{typical trans} satisfy the constraints and that the 
system matrices can be reconstructed. 
In this section, we provide two concrete procedures to achieve this goal.

%%%%%%%%%%%%%%%%%%%%%%%%%%%%%%%%%%%%%%%%%%%%%

\subsection{Reconstruction of system matrices}
\label{reconstruction from classical method}

Let $(A_0, B_0, C_0)$ be constructed from the transfer function 
of a minimal quantum system \eqref{dynamics} and \eqref{observation} 
(note that now it is not limited to the SISO case). 
Then, for a certain matrix $T$, the matrices \eqref{typical trans} 
satisfy the constraints \eqref{A matrix}, which immediately yields 
$A+A^\dagger+C^\dagger C=0$, and $B=-C^\dagger$. 
These conditions are written in terms of $(A_0, B_0, C_0)$ as 
\begin{equation}
\label{Lyapunov}
   (T^\dagger T)A_0 + A_0^\dagger (T^\dagger T) 
      + C_0^\dagger C_0 = 0
\end{equation}
and $(T^\dagger T)B_0=-C_0^\dagger$. 
Now the system is assumed to be minimal, thus $A_0$ is Hurwitz 
from Lemma~\ref{lemma 3.1}. 
This means that the Lyapunov equation \eqref{Lyapunov} has a 
unique solution $T^\dagger T>0$. 
Accordingly, we have the diagonalization 
$T^\dagger T=U_0\Lambda U_0^\dagger$, 
where $\Lambda>0$ is a diagonal matrix composed of eigenvalues 
of $T^\dagger T$ and $U_0$ the corresponding unitary matrix. 
Then, $T$ is fully characterized by an arbitrary unitary matrix 
$U$ as 
\begin{equation}
\label{T from classical}
    T=U\sqrt{\Lambda}U_0^\dagger, 
\end{equation}
where $\sqrt{\Lambda}$ is a positive diagonal matrix satisfying 
$(\sqrt{\Lambda})^2=\Lambda$. 
This $T$ generates the equivalence class of quantum systems. 
In particular, by denoting $T_0=\sqrt{\Lambda}U_0^\dagger$, we 
can interpret that $T$ first transforms the matrices $(A_0, B_0, C_0)$ to 
those corresponding to the quantum system, 
$(T_0A_0T_0^{-1}, T_0B_0, C_0T_0^{-1})$; 
then we obtain the unitary equivalence class by acting a unitary matrix 
$U$ on those matrices. 
See Fig.~\ref{unitary equivalence}. 
%Also note that in a practical situation we need to check if 
%\eqref{T from classical} satisfies the other condition 
%$(T^\dagger T)B_0=-C_0^\dagger$. 

Now the system matrices $(\Omega, C)$ can be reconstructed. 
It follows from \eqref{A matrix} that $A-A^\dagger=-2i\Omega$, 
which thus together with \eqref{typical trans} and \eqref{T from classical} 
yields
\begin{eqnarray}
\label{omega eq sec 4}
& & \hspace*{-3em}
    \Omega=U\Omega_0 U^\dagger,~~~
\nonumber \\ & & \hspace*{-3em}
    \Omega_0=\frac{i}{2}\Big[ 
       \sqrt{\Lambda} U_0^\dagger A_0 U_0 \sqrt{\Lambda^{-1}}
      -\sqrt{\Lambda^{-1}} U_0^\dagger A_0^\dagger U_0 \sqrt{\Lambda} \Big]. 
\end{eqnarray}
Similarly, from $C=C_0T^{-1}$ we have 
\begin{equation}
\label{c eq sec 4}
    C=(C_0 U_0 \sqrt{\Lambda^{-1}})U^\dagger. 
\end{equation}
These are exactly of the form \eqref{equivalent class transfer} 
in Theorem~\ref{equivalent class}. 
Hence, the following theorem holds. 
Note that a similar result is found in \cite{Petersen2011b}.

\begin{theorem}\label{th.5.1}
Let $A_0$ and $C_0$ be matrices directly obtained from the transfer 
function $\Xi(s)$, e.g. \eqref{typical choice} in the SISO case. 
Then, the equivalence class of system matrices $(\Omega, C)$ is given 
by \eqref{omega eq sec 4} and \eqref{c eq sec 4} with unitary matrix 
$U$, where $\Lambda$ and $U_0$ are constructed from the 
solution of \eqref{Lyapunov}. 
\end{theorem}

\begin{figure}
\centering
\includegraphics[scale=0.32]{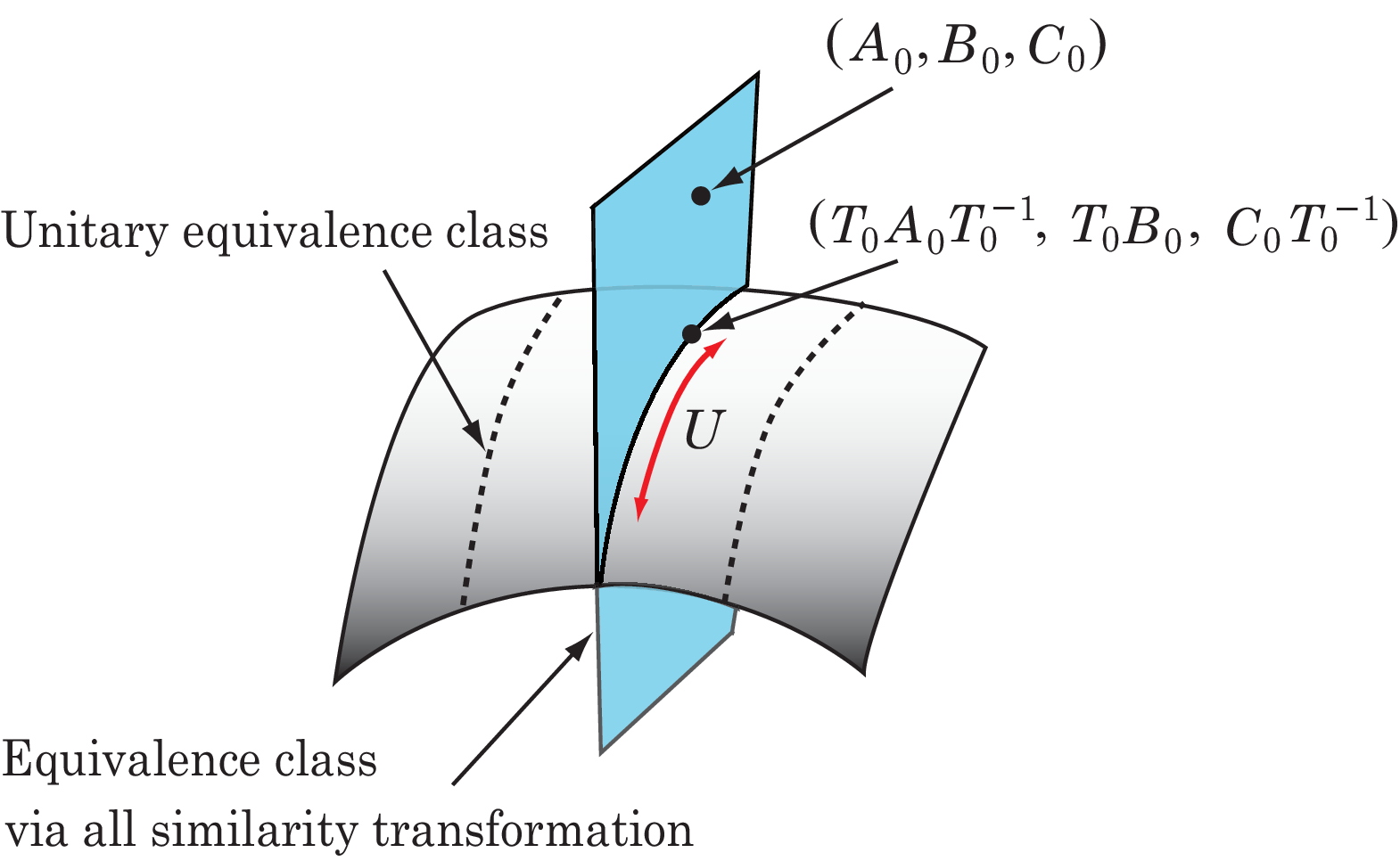}
\caption{
\label{unitary equivalence}
Unitary equivalence class of the system matrices, which is generated 
from $(A_0, B_0, C_0)$. 
We denote $T_0=\sqrt{\Lambda}U_0^\dagger$. 
}
\end{figure}

%%%%%%%%%%%%%%%%%%%%%%%%%%%%%%%%%%%%%%%%%%%

\subsection{Example}

Let us consider a two-mode SISO system with only single mode accessible 
and assume that the following transfer function has been experimentally 
obtained: 
\[
    \Xi(s) = 1 + \frac{c_1s}{s^2+a_1s+a_0}, 
\]
where $a_0, a_1>0$ and $c_1$ are real numbers. 
(As we will explain later, $c_1=-2a_1$ is satisfied.) 
For this transfer function we take the typical realization 
\eqref{typical choice}; i.e., 
\[
    A_0=\left[ \begin{array}{cc}
           0 & 1 \\
           -a_0 & -a_1 \\
        \end{array} \right],~~
    B_0=\left[ \begin{array}{c}
           0 \\
           1 \\
        \end{array} \right],~~
    C_0=[0,~c_1]. 
\]
Note that $B_0=-C_0^\dagger$ does not hold in general. 
With this choice, the Lyapunov equation \eqref{Lyapunov} 
has the following unique solution: 
\[
    T^\dagger T
       =\frac{c_1^2}{2a_1}
        \left[ \begin{array}{cc}
           a_0 & 0 \\
           0 & 1 \\
        \end{array} \right], 
\]
which is equal to $\Lambda$, and now $U_0=I$. 
Thus, the equivalence class of the system matrices are given by 
\eqref{omega eq sec 4} and \eqref{c eq sec 4} with 
\[
    \Omega_0
       =\left[ \begin{array}{cc}
           0 & i\sqrt{a_0} \\
           -i\sqrt{a_0} & 0 \\
        \end{array} \right],~~
    C_0 U_0 \sqrt{\Lambda^{-1}}=[0, -\sqrt{2a_1}]. 
\]
In particular, when choosing $U=[0,~-1~;~i,~0]$, we have 
\[
    \Omega
      = \left[ \begin{array}{cc}
           0 & \sqrt{a_0} \\
           \sqrt{a_0} & 0 \\
        \end{array} \right],~~
    C=[\sqrt{2a_1},~0], 
\]
which have exactly the same forms as the system matrices in Example~3.3 
with specifically $\theta_2=0$ taken. 
That is, the coupling strength between the system modes is identified as 
$\sqrt{a_0}$, and the system-field coupling strength is identified as 
$\sqrt{2a_1}$. 
Note that the condition $(T^\dagger T)B_0=-C_0^\dagger$ 
yields $c_1=-2a_1$; indeed this relation is satisfied for the two-mode 
system, as easily seen by again setting $\theta_2=0$ in Example~3.3.

%%%%%%%%%%%%%%%%%%%%%%%%%%%%%%%%%%%%%%%%%%%%

\subsection{Direct reconstruction of system matrices from the 
transfer function}

In Section \ref{reconstruction from classical method} we have shown that 
the equivalent class of system matrices can be reconstructed through typical 
realization methods employed in classical system theory. 
We here present another procedure that directly reconstructs the 
equivalence class.

We begin with the simple SISO model where the coupling matrix is 
of the form $C=(\sqrt{\theta}, 0, \ldots, 0)$ with $\theta>0$ an unknown 
parameter; 
that is, we assume that only a single mode is accessible. 
However, we do not assume a specific structure on $\Omega$ and 
write it as 
\begin{equation}
\label{def of Omega}
   \Omega
    =\left[ \begin{array}{cc}
       \Omega_{11} & E \\
       E^\dagger & \tilde\Omega \\
     \end{array} \right], 
\end{equation}
where $\tilde\Omega$ is a Hermitian matrix with dimension $n-1$, 
$\Omega_{11}$ is a real number, and $E$ is a $n-1$ dimensional 
complex column vector. 
In this case, the transfer function \eqref{Q transfer function} 
is given by 
\[
    \Xi(s)
     =1-\theta\Big(
           s+i\Omega_{11}+\frac{\theta}{2}
              +E(s+i\tilde\Omega)^{-1}E^\dagger
                 \Big)^{-1}. 
\]
Again we assume that $\Xi(s)$ is known. 
The parameters are then reconstructed as follows.

First, through a straightforward calculation we have 
\[
    s(1-\Xi(s))
     = \frac{\theta}
        {1+i\Omega_{11}/s + 1/2s + E(s^2+is\tilde\Omega)^{-1}E^\dagger}, 
\]
which thus leads to 
\[
   \theta=\lim_{|s|\rightarrow\infty}s(1-\Xi(s)). 
\]
Next, since now $\theta$ has been identified, we can further identify 
$\Omega_{11}$ using the following equation: 
\[
   \Omega_{11}
    = \lim_{|s|\rightarrow\infty}\Big[
         \frac{i\theta(\Xi(s)+1)}{2(\Xi(s)-1)} + is \Big]. 
\]
Now, $\theta$ and $\Omega_{11}$ have been obtained in addition 
to $\Xi(s)$. 
This means that  the function 
$\tilde{\Xi}(s):=E(sI+i\tilde\Omega)^{-1}E^\dagger$ is known. 
We diagonalize $\tilde\Omega$ as $\tilde\Omega=V\tilde\Lambda V^\dagger$ 
with 
$\tilde\Lambda
={\rm Diag}\{\tilde\lambda_1, \ldots, \tilde\lambda_{n-1}\}$. 
Then, $\tilde{\Xi}(s)=EV(sI-\tilde\Lambda)^{-1}(EV)^\dagger$ is 
of the form 
\[
    \tilde{\Xi}(s)=\sum_{i=1}^{n-1}\frac{|E'_i|^2}{s+i\tilde\lambda_i},
\]
where $E'_i$ is the $i$-th element of $EV$. 
This implies that $\tilde\lambda_i$ can be detected by examining 
the function $\tilde{\Xi}(i\omega)$; 
that is, $-i\tilde\lambda_i$ is the value on the imaginary axis 
such that $\tilde{\Xi}(i\omega)$ diverges. 
Then, (assuming that $\tilde{\Omega}$ has non-degenerate spectrum) 
we can further determine $|E'_i|^2$ from 
\[
    |E'_i|^2 = (s+i\tilde\lambda_i)\Xi(s)\big|_{s=-i\tilde\lambda_i}.
\]
Lastly, let us express $E'_i$ as $E'_i=e^{i\phi_i}|E'_i|$ with phase 
$\phi_i$ and define $\Phi={\rm Diag}\{\phi_1,\ldots,\phi_{n-1}\}$. 
Then, \eqref{def of Omega} can be written 
\[
   \Omega
    =\left[ \begin{array}{cc}
       1 & 0 \\
       0 & V e^{-i\Phi} \\
     \end{array} \right]
     \left[ \begin{array}{cc}
       \Omega_{11} & |E'| \\
       |E'|^\top & \tilde\Lambda \\
     \end{array} \right]
     \left[ \begin{array}{cc}
       1 & 0 \\
       0 & e^{i\Phi} V^\dagger \\
     \end{array} \right], 
\]
where $|E'|=[|E'_1|,\ldots,|E'_{n-1}|]$. 
As shown above, the middle matrix can be completely identified from 
the transfer function $\Xi(s)$. 
Therefore, all the eignevalues of $\Omega$ can now be determined. 
In the case when $\tilde{\Omega}$ is degenerated, all the elements of 
the vector $|E'|$ cannot be determined, but $\Omega_{11}$ and 
$\tilde{\Lambda}$ can be. 
Thus as in the above case the eigenvalues of $\Omega$ can be identified. 
Let us now summarize the result.

\begin{theorem}\label{th.5.2}
The equivalence class of systems having a given transfer function 
$\Xi(s)$ is completely parameterized by the set of parameters 
$(\theta, \Omega_{11}, |E_{i}^{\prime}|, \tilde{\lambda}_{i}) 
\in \mathbb{R}^{2n}$, which are directly computed from $\Xi(s)$ 
using the above procedure. 
In particular, the coupling parameter $\theta$ and the eigenvalues of 
$\Omega$ can be identified. 
\end{theorem}

To describe the general case, we assume that the $m\times n$ matrix 
$C$ is of rank $m$, meaning that all the injected input fields couple with 
the system. 
Furthermore, we assume $m\leq n$; 
in this case, without loss of generality, $C$ can be expressed as 
$C=(\tilde{C}, 0)$, with $\tilde{C}$ a $m\times m$ full rank complex matrix. 
Correspondingly, we represent $\Omega$ as in the same form 
\eqref{def of Omega}, in which case $\Omega_{11}$ is a 
$m\times m$ Hermitian matrix. 
Then, as in the previous case we have 
\[
    \tilde{C}\tilde{C}^\dagger
       =\lim_{|s|\rightarrow\infty}s(1-\Xi(s)). 
\]
This means that $\tilde{C}$ can be represented in terms of a 
known strictly positive matrix $\tilde{C}_0$ and an arbitrary unitary 
matrix $\tilde{U}$ as $\tilde{C}=\tilde{C}_0 \tilde{U}$. 
Moreover, 
\[
    \tilde{U}\Omega_{11}\tilde{U}^\dagger
     = \lim_{|s|\rightarrow\infty}\Big[
         -i\tilde{C}_0^\dagger (I-\Xi(s))^{-1}\tilde{C}_0 + isI \Big] 
           + \frac{i}{2}\tilde{C}_0^\dagger\tilde{C}_0, 
\]
which means that $\Omega_{11}$ can be determined up to the 
unitary rotation by $\tilde{U}$. 
Now, we are given 
\[
    \tilde{\Xi}(s)=\tilde{U}E(sI+i\tilde{\Omega})^{-1} 
                    (\tilde{U}E)^\dagger. 
\]
Hence, from the same procedure as in the simple case, we can determine 
the eigenvalues of $\tilde{\Omega}$ and $E_iE_j^*$ from $\tilde{\Xi}(s)$. 
Consequently, the eigenvalues of $\Omega$ can be also be reconstructed.

%%%%%%%%%%%%%%%%%%%%%%%%%%%%%%%%%%%%%%%%%%%
%%%%%%%%%%%%%%%%%%%%%%%%%%%%%%%%%%%%%%%%%%%
%%%%%%%%%%%%%%%%%%%%%%%%%%%%%%%%%%%%%%%%%%%

\section{Statistical analysis of the system identification problem}
\label{sec.statistics}

In this section we study the problem of \emph{how} to identify 
the unknown parameters of a linear system, and related questions such 
as which input states are optimal, what is the quantum Fisher information 
of the output, and which output measurements should be performed.

As before, we suppose that the system dynamics depends on an unknown 
parameter $\theta\in \Theta$, as $\Omega=\Omega(\theta)$ and 
$C = C(\theta)$. 
We will probe the system with a coherent input state $|\alpha(t)\rangle$ 
whose temporal profile is given by the complex amplitude function 
$\alpha(t)\in L^2( \mathbb{R}, \mathbb{C}^m)$. 
In experiments, $\alpha(t)$ would be supported in the finite time interval 
of the experiment, but for our analysis the time length will not be considered 
as an essential resource, but rather the total ``energy'' $E= \int |\alpha(t)|^2 dt$ 
used to excite the system. 
We will furthermore assume that the Fourier transform 
$\tilde{\alpha}(\omega)$ concentrates around a finite set of frequencies 
$\omega_1,\dots ,\omega_p$, so that in the frequency domain the input 
state can be approximated by the finite mode continuous variables state
$$
    | \vec{\bf z}, \vec{\omega} \rangle_{\rm in} \approx 
    |{\bf z}_{1};\omega_{1}\rangle\otimes\dots\otimes 
                              |{\bf z}_{p};\omega_{p}\rangle, 
$$
where $\vec{\bf z}:= ({\bf z}_{1}, \dots ,{\bf z}_{p})$, $\vec{\omega}
:= (\omega_{1}, \dots \omega_{p})$, and $|{\bf z}_{i} ;\omega_{i}\rangle$ 
represent the coherent state with amplitude ${\bf z}_{i}\in{\mathbb C}^m$ 
and frequency $\omega_i$. 
In this representation, the ``energy'' constraint is $E= \sum_i |{\bf z}_{i}|^2$.

Since the system is linear, the output is obtained by rotating the amplitude 
vector ${\bf z}$ by the $\theta$-dependent transfer function 
$ \Xi_{\theta}(-i\omega)$, separately for each frequency mode 
$$
   |{\bf z}_{i};\omega_{i}\rangle\longmapsto
   | \Xi_{\theta}(-i\omega_{i}) {\bf z}_i\rangle,
$$
so the the output state is
\begin{eqnarray}
& & \hspace*{-2em}
\label{eq.output.coherent}
    | \vec{\bf z}_\theta, \vec{\omega} \rangle_{\rm out} \approx 
    |\Xi_{\theta}(-i\omega_{1}) {\bf z}_{1};\omega_{1}\rangle
        \otimes\dots\otimes 
            |\Xi_{\theta}(-i\omega_{p}){\bf z}_{p};\omega_{p}\rangle.
\nonumber \\ & & \hspace*{-1.5em}
\mbox{}
\end{eqnarray}
The task is now to perform an appropriate measurement and provide 
an estimator $\tilde{\theta}$ of $\theta$ based on the measurement 
data. 
The parameter estimation for such ``unitary rotation'' families of states 
is a fairly well understood topic in quantum statistics \cite{Holevo}, 
but for reader's convenience we briefly recall some of the key concepts here.

For a quantum system with Hilbert space $\mathcal{H}$, an arbitrary 
measurement $M$ with values in the probability space $(\mathcal{X},\Sigma)$ 
is described by a positive operator valued measure (POVM) over 
$(\mathcal{X},\Sigma)$, i.e. a family $M:= \{m(A) : A\in \Sigma\}$ 
of operators on $\mathcal{H}$ satisfying the properties 
\begin{itemize}
\item
    positivity: $m(A)\geq 0$ for all events $A\in \Sigma$;
\item
    $\sigma$-additivity: for any disjoint countable family of events $A_i$, 
    $\sum_i m(A_i)= m(\cup_i A_i)$ holds;
\item
    normalization: $m(\mathcal{X})= \mathbf{1}$.
\end{itemize}
When the system is in state $\rho$, the probability distribution of the 
measurement outcome $X$ is $\mathbb{P}^M_\rho(dx) = {\rm Tr}(\rho m(dx))$. 
Now consider that the state depends on an unknown one-dimensional 
parameter $\theta\in \Theta\subset\mathbb{R}$, such that 
$\theta\mapsto \rho_\theta$ forms a smooth family of states. 
The multidimensional case will be discussed later. 
In order to estimate $\theta$ we perform a measurement $M$ and 
construct an estimator $\tilde{\theta}(X)$, whose performance can be 
measured by the mean square error (MSE)
$$
    \mathbb{E}_\theta \big[(\tilde{\theta}- \theta)^2 \big] 
        = \int \big(\tilde{\theta}(x)- \theta\big)^2\mathbb{P}^M_{\rho_\theta}(dx).
$$
As the MSE depends on the measurement and the chosen estimator, one 
would like to find an optimal procedure minimizing the MSE. 
The {\it quantum Cram\'{e}r-Rao bound} \cite{Braunstein&Caves} states that 
for any measurement and any unbiased estimator $\tilde{\theta}$ 
(i.e. $\mathbb{E}_\theta(\tilde{\theta}) = \theta$) the following lower 
bound holds:
\begin{equation}
\label{eq.QCR}
    \mathbb{E}_\theta \big[  (\tilde{\theta}- \theta)^2 \big] \geq  F(\theta)^{-1},
\end{equation}
where $F(\theta)={\rm Tr}(\rho_\theta L_\theta^2)$ is the 
\emph{quantum Fisher information} (QFI) and $L_\theta=L_\theta^\dagger$ 
is the symmetric logarithmic derivative defined through the operator-valued 
equation 
$$
    \frac{d\rho_\theta}{d\theta} 
      = \frac{1}{2}(L_\theta\rho_\theta + \rho_\theta L_\theta).
$$
In particular, if $\rho_\theta= \ket{\psi_\theta}\bra{\psi_\theta}$ is 
a pure state family, then 
\begin{equation}
\label{pure Fisher}
     F(|\psi_\theta\rangle) 
        = 4 \Big(\pro{\psi'_\theta}{\psi'_\theta} 
                   - |\pro{\psi'_\theta}{\psi_\theta}|^2 \Big), 
\end{equation}
where $\ket{\psi'_\theta}=d\ket{\psi_\theta}/d\theta$.

%with $|\psi_\theta\rangle= \exp(i\theta G)|\psi\rangle$ having 
%self-adjoint generator $G$, then 
%\begin{equation}\label{fisher.variance}
%F(\rho_\theta) = 4{\rm Var}(G):=  4 (\langle G^2 \rangle - \langle G\rangle^2)
%\end{equation}
%where $\langle \cdot \rangle$ denotes the expectation with respect to 
%$|\psi\rangle$. 
The bound \eqref{eq.QCR} is achievable when a large number $n$ of 
copies of $\rho_\theta$, in the sense that there exist measurements 
and estimators $\tilde{\theta}_n$ such that
$$ 
      \lim_{n\to\infty} n \cdot \mathbb{E}_\theta[  (\tilde{\theta}_n- \theta)^2 ] 
          = F(\rho_\theta)^{-1}.
$$
In our case that $\ket{\vec{\bf z}_\theta, \vec{\omega}}_{\rm out}$ is 
a product of independent coherent states, each frequency mode $\omega_i$ 
carries an amount of QFI which is proportional to the change of the amplitude 
$\Xi_{\theta}(-i\omega_{i}) {\bf z}_i$ with $\theta$. 
The total QFI is given by the following convex combination of individual 
informations: 
\[     
    F(\theta) = \sum_{i=1}^p F_i(\theta)
                 = 4 E \cdot
        \sum_{i=1}^p \frac{\| {\bf z}_i\|^2}{E} \left\| 
                          \frac{d\Xi_{\theta}(-i\omega_{i})}
                                   {d\theta}\frac{ {\bf z}_i}{\|{\bf z}_i\|}\right\|^2.
\]
This implies that, for a one-dimensional parameter, the optimal input 
consists of a coherent signal with single frequency $\omega_{\rm opt}$ 
and amplitude ${\bf z}_{\rm opt} = E {\bf w}_{\rm opt}$ defined as the 
solution of the following optimization problem:
\begin{equation}
\label{freq optim}
    (\omega_{\rm opt}, {\bf w}_{\rm opt}) 
        = \underset{\omega, \|{\bf w}\|=1 }{\arg\max}
            \left\| \frac{d\Xi_{\theta}(-i\omega)}{d\theta} {\bf w}\right\|^2.
\end{equation} 
As $\Xi_{\theta}(-i\omega)$ is unitary, the generator 
$G_{\theta}= i d \Xi_{\theta}(-i\omega)/d\theta$ is self-adjoint. 
Thus ${\bf z}_{\rm opt}$ is given by the eigenvector 
of $G_{\theta}$ whose eigenvalue has the largest absolute value. 
Then the optimal QFI is
\begin{equation}
F_{\rm opt} 
        = 4E  \underset{\omega}{\max}
            \left\| \frac{d\Xi_{\theta}(-i\omega)}{d\theta} \right\|^2,
\end{equation} 
and it can be achieved asymptotically by performing adaptive homodyne 
measurements \cite{Wiseman adaptive}.

%%%%%%%%%%%%%%%%%%%%%%%%%%%%%%%%%%%%%%%%%%%%

\subsection{SISO example}

Consider the single mode (i.e. $n=1$) SISO system with parameters 
$\Omega = \theta_1$ and $C= \theta_2$, such as an ideal mechanical 
oscillator with resonant frequency $\theta_1$. 
The transfer function is then 
\begin{eqnarray}
& & \hspace*{-2em}
\label{Estimation SISO example}
  \Xi_{\theta}(-i\omega) 
     = \frac{-i\omega + i\theta_1 - \theta_2^2/2}
                { -i\omega + i\theta_1 +\theta_2^2 /2} 
     = -\exp(-2i\phi(\omega, \theta_1, \theta_2)),
\nonumber \\ & & \hspace*{-1.5em}
\mbox{}
\end{eqnarray}
where 
$$
  \phi(\omega, \theta_1, \theta_2) 
     = \arctan \Big( \frac{-2\omega+2\theta_1}{\theta_2^2} \Big)
$$
is the phase of $i(-\omega +\theta_1) - \theta_2^2/2$.
We distinguish two cases depending on which of $\theta_1$ and $\theta_2$ 
is considered to be unknown.

If $\theta_1$ is unknown, then QFI at frequency $\omega$ is given by 
$$
     F(\theta_1;\omega) 
        = 16E \left| \frac{d\phi(\omega, \theta_1, \theta_2)}{d\theta_1} 
                       \right|^2 
        = 16 E\left| \frac{2\theta_2^2 }{\theta_2^4 + 4 (\omega-\theta_1)^2}
                       \right|^2.
$$
This takes the maximum $F_{\rm opt} = 64 E \theta_2^{-4}$ at 
$\omega_{\rm opt}=\theta_1$. 
There are three remarks on this result.

Firstly, $\omega_{\rm opt}=\theta_1$ means that the optimal input is a coherent 
field with {\it unknown} resonant frequency. 
In practice, one can adopt an adaptive strategy whereby one initially injects 
a signal composed of sufficiently many frequencies, also called ``M-sequence" 
\cite{LjungBook}, followed by more precise inputs targeting the optimal 
frequency. 
Secondly, the optimal QFI $F_{\rm opt} = 64 E \theta_2^{-4}$ increases as 
$\theta_2$ decreases and the system becomes less stable (note that the 
system's $A$ matrix has eigenvalue $-i\theta_1 - \theta_2^2/2$). 
This is expected due to the longer coherence time, but it also implies that
the time to reach the asymptotic regime is longer.
Therefore, as in the classical case, there exists a trade-off between the 
stability and the information for system identification. 
The third observation is that the maximum QFI $F_{\rm opt}$ 
can be achieved for large ${\bf z}$ by adaptively choosing the optimal frequency, 
and by performing a homodyne measurement of an appropriate quadrature, similar 
to the adaptive phase estimation protocol of \cite{Wiseman adaptive}.

We pass now to the second case where $\theta_2$ is unknown. 
In this case, QFI at frequency $\omega$ is 
$$
F(\theta_2;\omega) 
    = 16E \left| \frac{d\phi(\omega, \theta_1, \theta_2)}{d\theta_2} \right|^2 
    = 16 E\left| \frac{4(-\omega+\theta_1)\theta_2 }
                                  {\theta_2^4 + 4 (-\omega+\theta_1)^2}\right|^2.
$$
By optimizing over $\omega$ we find that the largest QFI is 
achieved at $\omega_{\rm opt}  = \theta_1 \pm \theta_2^2/2 $ and 
is equal to $F_{\rm opt} = 16E\theta_2^{-2}$. 
Note that in this case $F_{\rm opt}$ depends on the unknown parameter 
$\theta_2$.

Similar techniques can be applied to the more general case of one-dimensional parameters.
For instance, a SISO passive linear system can be represented as a 
cascaded network of single-mode oscillators, hence the transfer function at 
$-i\omega$ is the complex phase \cite{Nurdin2010}
%
%\[
%    \Xi_{\theta}(-i\omega)
%      = (-1)^n \frac{(  \overline{- i\omega - \zeta_1}) }{ (- i\omega - \zeta_1)} \dots 
%        \frac{ ( \overline{ -i\omega - \zeta_n} )}{  ( -i\omega - \zeta_n)}
%      = (-1)^n\exp\Big(-2i \sum_j \arg (-i\omega - \zeta_j)  \Big).
%\]
%
% For double column format:
%
\begin{eqnarray*}
    \Xi_{\theta}(i\omega)
      &=& 
       (-1)^n \frac{(  \overline{ -i\omega - \zeta_1}) }{ ( -i\omega - \zeta_1)} \dots 
        \frac{ ( \overline{ -i\omega - \zeta_n} )}{  ( -i\omega - \zeta_n)}
    \\ &=& 
        (-1)^n\exp\Big(-2i \sum_j \arg (-i\omega - \zeta_j)  \Big).
\end{eqnarray*}
$\zeta_j$ is the $\theta$-dependent pole of the transfer function. 
In principle the optimal frequency can be obtained in the same way as above 
by maximizing QFI 
$F(\omega)= 4| d\Xi_{\theta}(-i\omega)/d\theta|^2$ over $\omega$.

%%%%%%%%%%%%%%%%%%%%%%%%%%%%%%%%%%%%%%%%%%%
\subsection{Estimation for multidimensional parameters}
%%%%%%%%%%%%%%%%%%%%%%%%%%%%%%%%%%%%%%%%%%%

\begin{figure}
\centering
\includegraphics[scale=0.35]{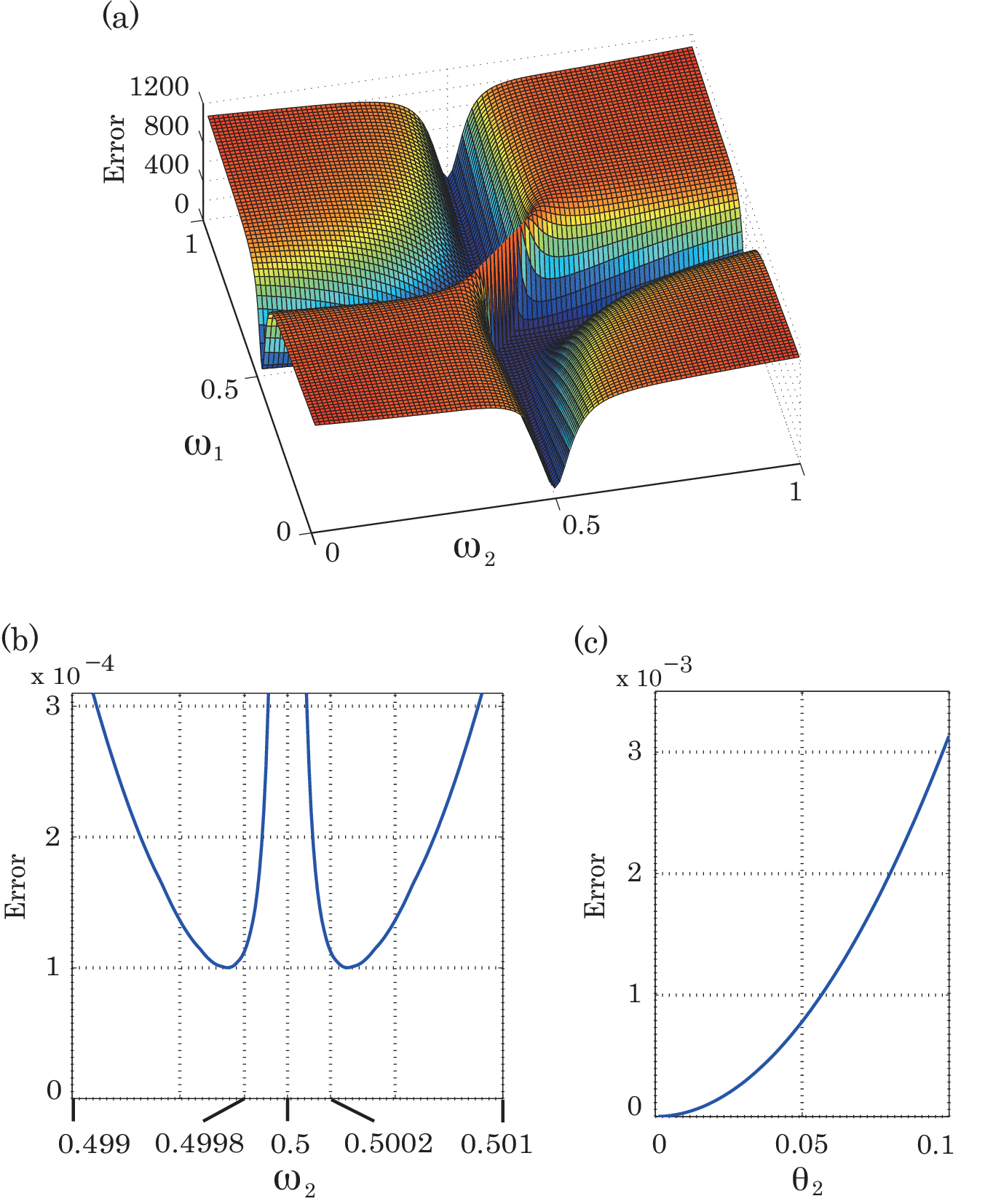}
\caption{
\label{Simulation}
(a) The lower bound of the total estimation error as a function of the 
frequencies $(\omega_1, \omega_2)$, in the case $\theta_1=0.5$ and 
$\theta_2=0.02$. 
(b) A cut through the previous plot at $\omega= \theta_1= 0.5$ shows 
two local minima at $\omega_2\approx \theta_1\pm \theta_2^2/2$. 
(c) Achievable lower bound for the MSE as a function of $\theta_2$, for 
the values of $\omega_{1,2}$ described above and with $r=1/2$. 
}
\end{figure}

The theory for one-dimensional parameter can be extended 
to multi-dimensional parameters 
$\theta=[\theta_1, \ldots, \theta_m]^T\in \mathbb{R}^{m}$. 
In this case the error covariance matrix is bounded by the following 
Cram\'{e}r-Rao matrix inequality:
\begin{equation}
\label{multi QCR}
    \mathbb{E}_\theta \big[  (\tilde{\theta}- \theta)(\tilde{\theta}- \theta)^T \big] 
       \geq  F^c(\theta)^{-1} \geq  F(\theta)^{-1}.
\end{equation}
$\tilde{\theta}$ is the vector of unbiased estimators. 
$F^c(\theta)$ is the {\it classical} Fisher information (CFI) 
matrix corresponding to the probability distribution of a particular 
measurement process, while $F(\theta)$ is the QFI matrix of the 
output state, defined similarly to the one dimensional case 
\cite{Holevo,Braunstein&Caves}.

However, the quantum Cram\'e{r}-Rao bound is in general not achievable 
due to incompatibility of the optimal measurements corresponding to 
different parameter components. 
We will therefore focus on the possibly sub-optimal setup where a dual 
homodyne (heterodyne) measurement is performed on each output mode. 
Essentially this means that the output is split into two channels, and 
complementary quadratures are measured on each. 
In particular, this implies that the MSE for the heterodyne measurement 
is at most a factor two larger than that of the optimal measurement. 
For a one-mode coherent state $\ket{z}$ the probability density of the 
measurement outcome is the two-dimensional Gaussian centered at 
$({\Re}(z), {\Im}(z))$ and variance equal to two times the vacuum fluctuations: 
$p(y)={\cal N}({\Re}(z), {\Im}(z), \mathbf{1})$.

As an example, we consider the same SISO system as above, but in this 
case the unknown parameter is $\theta= (\theta_1,\theta_2)$. 
We will consider an input consisting of several frequencies, with 
corresponding output amplitudes 
${\bf z}_{\theta;i} =  \Xi_{\theta}(-i\omega_{i}) {\bf z}_i\in \mathbb{C}$ , 
for $i=1,\dots ,p$. 
%In general for a coherent state $\ket{\alpha_\theta}$ with two parameters 
%$\theta=(\theta_1, \theta_2)$, a repeated dual Homodyne detection provides 
%the classical probability density 
%$p(y ; \theta)={\cal N}({\Re}(\alpha_\theta)){\cal N}({\Im}(\alpha_\theta))$, 
%where ${\cal N}(\alpha)$ is the normal distribution with mean $\alpha$ and 
%variance $1/2$. 
The $jk$ element of the CFI matrix of $p(y ; \theta)$ is then given by 
%
%\[
%     F^c_{jk}(\theta) = 
%     E \sum_{i=1}^{p} f^{c}_{jk,i}(\theta) =  E \cdot \sum_{i=1}^{p} 
%                  \frac{2|{\bf z}_{i}|^{2}}{E} 
%                  \left[ \frac{\partial {\Re}({\bf z}_{\theta;i})}{\partial \theta_j}
%                        \frac{\partial {\Re}({\bf z}_{\theta;i})}{\partial \theta_k}
%                   + \frac{\partial {\Im}({\bf z}_{\theta;i}) }{\partial \theta_j}
%                        \frac{\partial {\Im}({\bf z}_{\theta;i})}{\partial \theta_k}\right].
%\]
%
% For double column format:
%
\begin{eqnarray}
& & \hspace*{-2em}
     F^c_{jk}(\theta) = 
     E \sum_{i=1}^{p} f^{c}_{jk,i}(\theta) 
\nonumber \\ & & \hspace*{-1.5em}
     =  E \cdot \sum_{i=1}^{p} 
                  \frac{2|{\bf z}_{i}|^{2}}{E} 
                  \left[ \frac{\partial {\Re}({\bf z}_{\theta;i})}{\partial \theta_j}
                        \frac{\partial {\Re}({\bf z}_{\theta;i})}{\partial \theta_k}
                   + \frac{\partial {\Im}({\bf z}_{\theta;i}) }{\partial \theta_j}
                        \frac{\partial {\Im}({\bf z}_{\theta;i})}{\partial \theta_k}\right].
\nonumber
\end{eqnarray}
%
%The reason why we chose this measurement is that it is indeed optimal 
%in the sense that the above $F^c(\theta)$ is identical to the QFI matrix \cite{}. 
%Therefore we can evaluate the optimal input and the attainable lower 
The explicit expression of the (normalized) CFI matrix is
%
%\[
%     f^c(\theta ; \omega)
%       = \frac{8}{( (\theta_1-\omega)^2 + \theta_2^4/4 )^2}
%          \left[ \begin{array}{cc}
%               \theta_2^4/4  & -(\theta_1-\omega)\theta_2^3/2 \\
%               -(\theta_1-\omega)\theta_2^3/2  &  (\theta_1-\omega)^2\theta_2^2 \\
%          \end{array} \right].
%\]
%
% For double column format:
%
\begin{eqnarray}
& & \hspace*{-2em}
      f^c(\theta ; \omega)
\nonumber \\ & & \hspace*{-1.5em}
       = \frac{8}{( (\omega - \theta_1)^2 + \theta_2^4/4 )^2}
       \left[ \begin{array}{cc}
               \theta_2^4/4  & (\omega - \theta_1)\theta_2^3/2 \\
               (\omega - \theta_1)\theta_2^3/2  &  (\omega - \theta_1)^2\theta_2^2 \\
          \end{array} \right].
\nonumber
\end{eqnarray}
Note that ${\rm rank} (f^c(\theta ; \omega))=1$, which simply means that 
a single coherent input state with fixed $\omega$ can only identify one 
component of the parameter. 
We will therefore consider the case of two frequency modes $\omega_1$ and 
$\omega_2$.
By asymptotic efficiency theory, the MSE ${\mathbb E}_\theta[(\tilde{\theta}_1
-\theta_1)^2 + (\tilde{\theta}_2-\theta_2)^2]$ of optimal estimators (e.g. the 
maximum likelihood) scales as $\epsilon/E$ where
%
%\[
%    \epsilon = {\rm trace}\left[f^{c}(\theta)^{-1}\right]=
%    {\rm trace}\Big[
%    \Big( r f^c(\theta ; \omega_1) + (1-r) f^c(\theta ; \omega_2) \Big)^{-1} \Big], 
%\]
%
% For double column format:
%
\begin{eqnarray}
& & \hspace*{-2em}
       \epsilon = {\rm trace}\left[f^{c}(\theta)^{-1}\right]
\nonumber \\ & & \hspace*{-1.3em}
     =
      {\rm trace}\Big[
         \Big( r f^c(\theta ; \omega_1) + (1-r) f^c(\theta ; \omega_2) \Big)^{-1} \Big], 
\nonumber
\end{eqnarray}
and $0<r<1$ is the weight of the input with frequency $\omega_1$. 
To find the optimal procedure and MSE one has to minimize $\epsilon$ over 
$r$ and $(\omega_{1},\omega_{2})$. 
Figure \ref{Simulation} (a) illustrates the dependence of $\epsilon$ on the 
frequencies $\omega_{1},\omega_{2}$, for a set of true 
parameters $\theta_1=0.5$ and $\theta_2=0.02$, where $r$ is optimized 
at each point. 
We find the values of the optimal frequencies are very near to those 
which were shown to be optimal in the two one-dimensional estimation 
problems, namely  $\omega_1\approx \theta_{1}$, and 
$\omega_2 \approx\theta_1\pm \theta_2^2/2$, cf. Fig.~\ref{Simulation} (b). 
For these values, and with $r=1/2$ the bound $\epsilon$ is given by 
$$
    \epsilon(\theta_{2}) = \frac{\theta_{2}^{2}}{16}(5+ \theta_{2}^{2}),
$$
which is plotted in Fig.~\ref{Simulation} (c). We note that as before, the 
MSE vanishes when the coupling constant 
$\theta_{2}$ goes to zero, and does not depend on $\theta_{1}$.

%As in the single parameter case, the MSE becomes smaller when the decay rate $\theta_2$ 
%is small, and in this case $\epsilon_{\rm opt}$ is independent to $\theta_1$. 
%But if $\theta_2$ takes a relatively large value (i.e. the system decays fast), 
%then $\epsilon_{\rm opt}$ becomes bigger, and it is going to be dependent to 
%$\theta_1$. 

%%%%%%%%%%%%%%%%%%%%%%%%%%%%%%%%%%%%%%%%%%%
\subsection{Heisenberg scaling}
%%%%%%%%%%%%%%%%%%%%%%%%%%%%%%%%%%%%%%%%%%%

The coherent input setup is fairly close to that of classical linear system 
identification. 
We will show now that the superposition principle allows us to attain higher 
estimation precision as encountered in quantum enhanced metrology 
\cite{Giovanetti}. 
Consider as above, a single-mode SISO model with unknown 
Hamiltonian $\Omega= \theta$ and known coupling $C= c$. 
Let the input field state be the \emph{coherent superposition} of the vacuum 
and the $n$-photon state of frequency $\omega$: 
$$
    | \psi\rangle_{\rm in} 
       = \frac{1}{\sqrt{2}} \left( |0 \rangle + | n ;\omega \rangle \right), 
$$
whose mean energy is $E= n/2$. We note that $\ket{n;\omega}$ is a state 
of the light field with continuous-mode $\hat b(t)$ satisfying 
\eqref{eq.comm.relations}, and refer to the Appendix for more details.

Now the system interacts with the field with initial state $\ket{\psi}_{\rm in}$. 
For times which are significantly longer than the duration of the input pulse, 
the system returns to the ground state due to the stability of the dynamics 
while the field state is transformed by the action of the transfer function, 
and the two are decoupled from each other. 
In particular, the field output state is given by 
%
%\[
%   |\psi_\theta \rangle_{\rm out} 
%      = \frac{1}{\sqrt{2}} 
%          \big( |0\rangle + \Xi_\theta(-i\omega)^n| n ;\omega \rangle \big)
%      = \frac{1}{\sqrt{2}} 
%            \big( |0\rangle 
%                     + e^{-2 i n \phi(\omega, \theta,c)} | n ;\omega \rangle \big).
%\]
%
% For double colomn format
%
\begin{eqnarray}
\label{n photon output}
| \psi_\theta \rangle_{\rm out} 
    &=& \frac{1}{\sqrt{2}} 
            \left( |0\rangle + \Xi_\theta(-i\omega)^n| n ;\omega \rangle \right)
\nonumber \\
    &=& \frac{1}{\sqrt{2}} 
            \left( |0\rangle 
                     + e^{-2 i n \phi(\omega, \theta,c)} | n ;\omega \rangle \right).
\end{eqnarray}
For derivation, see Appendix. 
The QFI of $|\psi_\theta \rangle_{\rm out}$ is calculated as 
$$
   F(\theta) 
           = 16E^2 \left| \frac{d\phi(\omega, \theta, c)}{d\theta} \right|^2, 
$$
which is exactly the same as in the coherent input case, with the important 
difference that it has a quadratic (Heisenberg) scaling with $E$, familiar 
from quantum metrology models. 
In particular, the optimal frequency is $\omega_{\rm opt}=\theta_1$, and 
the corresponding QFI is $64 E^{2}/\theta_{2}^{4}$.  
As discussed before, since $\omega_{\rm opt}$ is unknown, in practice we 
can use an adaptive strategy in which the input frequency is repeatedly 
tuned to approach $\omega_{\rm opt}$ as the estimator becomes more and 
more accurate. 
Note however that the quadratic scaling with $E$ does not rely on the 
frequency distribution of the input, but rather on the ability to prepare 
superpositions of states with very different photon numbers. 
In particular, more realistic input signal containing a continuum of 
frequencies can achieve a similar scaling in $E$.

The  above input state is by no means the only design exhibiting quadratic 
scaling in $E$. 
Other schemes based on squeezed or  NOON states have been extensively 
discussed in the literature on quantum metrology \cite{Rafal}. 
%and we refer to the recent review \cite{Rafal} for more details on this topic. 
Here  we limit ourselves to listing some of the issues that require a more 
in depth analysis. 
The first question is whether the Heisenberg scaling can be achieved by performing realistic measurements, e.g. homodyne or photon counting. This question can be addressed by using the interferometric setup described in \cite{Caves}, which involves a product of squeezed and coherent input states. The optimization over input frequencies and general linear output measurements  can be formulated along the lines of the previous section, and will be addressed in a future publication. Other issues which have not been addressed are decoherence due to losses, and measurement imperfections. To some extent these can be modeled by extending the linear setup to include additional input-output channels which are not monitored.

%But now the M-sequence is generated on a different physical state other 
%than the coherent state. 
%That is, unlike the classical case, a physical state of the input signal for 
%identification should be carefully designed. 

\section{General linear systems}

In this paper we dealt with passive systems, as a special, but important class 
of linear input-output systems. 
We showed that taking this prior information into account leads to smaller 
equivalence classes than it is expected based on the classical theory. 
Additionally, in this case, the statistical estimation problem can be cast into 
that of optimizing the mean square error for a given energy of the input. 
For completeness, we will now sketch the general set-up of the system 
identification problem for linear systems which will be analysed in more 
detail elsewhere. We will use the following ``doubled-up" notation convention 
introduced in \cite{GoughPRA2010}. For a vector of operators 
$\hat{\bf x} = [\hat{x}_1, \dots ,\hat{x}_n]^T$ we denote 
$\breve{\bf x}:=[\hat{x}_1,\ldots,\hat{x}_n, \hat{x}^*_1,\ldots,\hat{x}^*_n]^{T}$. 
Given a linear transformation of the form 
$\hat{\bf y} = E_{-} \hat{\bf x} + E_{+} \hat{\bf x}^*$, we write
$$
\breve{\bf y} = 
\left[
\begin{array}{c}
 \hat{\bf y}  \\
  \hat{\bf y}^* \\
     \end{array} 
\right]
=
\Delta(E_{-}, E_{+}) \breve{\bf x}
:=
\left[ \begin{array}{cc}
    E_{-}   & E_{+} \\
     E_{+}^*  &  E_{-}^* \\
     \end{array} \right] 
\left[
\begin{array}{c}
 \hat{\bf x}  \\
  \hat{\bf x}^* \\
     \end{array} 
\right] ,  
$$
where $E_{-}^*, E_{+}^*$ denote the complex conjugates of the matrices 
$E_{-}, E_{+}$. 
For a $2n\times 2n$ matrix $X$ we define the {\it involution} 
$X^\flat = J^{(n)}X^\dagger J^{(n)}$ where 
$$
J^{(n)}:= \left[ \begin{array}{cc}
    I_n  & 0 \\
     0   &  - I_n \\
     \end{array} \right].  
$$
The $2n\times 2n$ matrix $\widetilde{S}$ is called $\flat$-unitary if 
$SS^\flat=S^\flat S$. 
The \emph{symplectic group} 
%${\rm Sp}(2n,\mathbb{C})$ 
is the subgroup of $\flat$-unitaries of the form $S= \Delta (S_{-}, S_{+})$ 
with $S_\pm$ suitable $n\times n$ complex matrices. 
Moreover, any $n\times n$ unitary $U$ can be identified with the 
``doubled-up" element $\widetilde{U}= \Delta(U, 0)$ of the symplectic 
group, so the unitary group can be seen as a subgroup of the symplectic one. 

%\vspace{2mm}

In order to describe the input-output relations for active systems we collect 
all of the system's variables into the vector 
$\breve{\veca}:=[\hat{a}_1,\ldots,\hat{a}_n, \hat{a}^*_1,\ldots,\hat{a}^*_n]^{T}$, 
which satisfies the commutation relations
$
[\breve{a}_i, \breve{a}_j^*] = J_{ij}.
$ 
For any symplectic matrix $S= \Delta (S_{-}, S_{+})$, there exists 
a Bogolubov transformation 
$\hat{\veca}^\prime = S_{-} \hat{\veca} + S_{+} \hat{\veca}^*$ which has 
the property that it preserves the above commutation relations. 
The system has a quadratic Hamiltonian of the form 
$$
\hat{H}= 
\breve{\veca}^\dagger \widetilde{\Omega} \breve{\veca}
%\hat{\bf a}^\dagger\Omega_{-} \hat{\bf a}  
%+ \hat{\bf a}^\dagger \Omega \hat{\bf a}^*
%+\hat{\bf a}^T R^{\dagger} \hat{\bf a}. 
$$
where $\widetilde{\Omega}:= -i \Delta(i\Omega_{-}, i\Omega_{+})$ is the 
generator of a symplectic transformation, 
i.e. $\exp(i\widetilde{\Omega})$ is a $\flat$-unitary. 
Equivalently, $\widetilde{\Omega} = \widetilde{\Omega}^\flat$, which 
means that the $n\times n$ matrices $\Omega_\pm$ satisfy the following 
conditions: $\Omega_{-}=\Omega_{-}^\dagger$ and $\Omega_{+}= \Omega_{+}^T$.  
The input $\hat{\bf B}(t)$ couples with the system through the operator
${L} = {C}_{-} \hat{\veca} + {C}_{+} \hat{\veca}^*$, where 
$ {C}_{-},  {C}_{+}$ are complex $m\times n$ matrices. 
In the Laplace domain, the input-output relations are given by 
\cite{GoughPRA2010}
$$
\mathcal{L}[\breve{\vecb}^{\rm out}] (s) 
= \widetilde{\Sigma} (s) \mathcal{L}[\breve{\vecb}] (s) 
$$ 
where $\widetilde{\Sigma} (s)$ is the transfer function 
\begin{eqnarray}
& & \hspace*{-2em}
\label{eq.transfer.fct.active}
\widetilde{\Sigma} (s):= 
 \left[ \begin{array}{cc}
    \Sigma_{-}(s) &  \Sigma_{+}(s) \\
      \Sigma_{+}(s^*)^*   &   \Sigma_{-}(s^*)^* \\
     \end{array} \right] = 
     I - \widetilde{C}(sI - \widetilde{A})^{-1} \widetilde{C}^\flat, 
\nonumber \\ & & \hspace*{-1.3em}
\mbox{}
\end{eqnarray}
with $\widetilde{C}:= \Delta(C_{-}, C_{+})$, and 
$\widetilde{A} := \Delta(A_{-}, A_{+})$, and 
$A_{\mp}:= -i\Omega_{\mp} - (C_{-}^\dagger C_\mp- C_{+}^T C_\pm^*)/2$.

As in the passive case, we would like to answer the following questions: 
what are the equivalence classes of dynamical parameters 
$(\widetilde{\Omega}, \widetilde{C})$ which have the same transfer function, 
and how can we estimate the identifiable parameters? 
Concerning the first question, we note that for any symplectic transformation 
$S$, the system with parameters $\Omega^\prime = S\widetilde{\Omega}S^\flat$ 
and $\widetilde{C}^\prime:= \widetilde{C}S^\flat$ has the  same transfer 
function \eqref{eq.transfer.fct.active}, and therefore all such parameters 
belong to the same equivalence class. 
As expected, the equivalence classes of general linear systems are larger than 
those of passive systems, since $n\times n$ unitaries are a subgroup of the 
symplectic group. 
We conjecture that the equivalence class is in fact completely determined 
by symplectic transformations, but this question will be addressed elsewhere.

Concerning the second question, we note that the active case differs from 
the passive one in some important respects, which are closely related to 
presence of squeezing elements in the dynamics. 
For instance, even if the input is in the vacuum state, the system's and 
output's stationary states may be mixed squeezed Gaussian states, and 
the two quantum systems may share quantum correlations. 
Although this makes the statistical analysis of the output state more involved, 
we expect that the tools developed for estimation of Gaussian states can be 
used to compute the quantum Fisher information of the output in terms of 
the transfer function, and to study the optimal input problem along the 
lines of the passive systems case.

%%%%%%%%%%%%%%%%%%%%%%%%%%%%%%%%%%%%%%%%%%%
\section{Conclusion and future works}
%%%%%%%%%%%%%%%%%%%%%%%%%%%%%%%%%%%%%%%%%%%

In Theorem \ref{equivalent class} we characterized the equivalence 
classes of linear input-output systems; 
minimal passive linear systems with the same transfer function are 
related by unitary transformations acting on the space of modes. 
Theorem \ref{infection theorem} states that systems satisfying the 
infection property are completely identifiable. 
Additionally, in Theorems \ref{th.5.1} and \ref{th.5.2} we provided two 
methods for finding the identifiable parameters and physical realizations 
for a given transfer function. 
We then addressed the statistical aspects of the system identification 
problem, and investigated the question of finding optimal input design 
and output measurement. 
The analysis is based on the statistical concepts of quantum and classical 
Fisher information. 
While for coherent inputs, the estimation error scales with the energy 
$E$ as $1/\sqrt{E}$, we showed that using non-classical input states 
we can attain the Heisenberg scaling $1/E$ due to the unitarity of the 
transfer function.

%Another important issue which was addressed only briefly here is how to actually estimate the system parameters. More precisely, what kind of input state should be chosen and what measurement should be performed on the output. The performance of such a design of experiment can be measured by using statistical tools such as Fisher information and asymptotic normality \cite{Guta2011,Catana&vanHorssen&Guta}.

There are a number of direction in which this work can be extended. For instance, in control applications it may be relevant to identify  physical realizations which optimize  the prediction rather than the estimation error. Since for large networks the identification becomes intractable, it may be useful to develop new system identification methods inspired by quantum compressed sensing \cite{Gross} and dimensional reduction.   Switching from passive to active linear systems, we conjectured that the equivalence classes consist of systems related by symplectic rather than unitary transformations. 
The system identification problem can be considered in a different setting, where the input fields are stationary (quantum noise) but have a non-trivial covariance matrix (squeezing). In this case the characterization of the equivalence classes boils down to finding the  systems with the same power spectral density, a problem which is well understood in the classical setting \cite{Anderson1969} but not yet addressed in the quantum domain.

%%%%%%%%%%%%%%%%%%%%%%%%%%%%%%%%%%%%%%%%%%%
%%%%%%%%%%%%%%%%%%%%%%%%%%%%%%%%%%%%%%%%%%%
%%%%%%%%%%%%%%%%%%%%%%%%%%%%%%%%%%%%%%%%%%%

\section*{Acknowledgment}

M.G.'s work was supported by the EPSRC grant  EP/J009776/1. 
N.Y.'s work was supported by JSPS Grant-in-Aid No. 24760341. 
Both authors are grateful for the hospitality of the Isaac Newton Institute 
for Mathematical Sciences, Cambridge, where this work was completed during 
the Quantum Control Engineering meeting.

%%%%%%%%%%%%%%%%%%%%%%%%%%%%%%%%%%%%%%%%%%%
%%%%%%%%%%%%%%%%%%%%%%%%%%%%%%%%%%%%%%%%%%%
%%%%%%%%%%%%%%%%%%%%%%%%%%%%%%%%%%%%%%%%%%%

\section*{Appendix}

A single photon (field) state is defined by 
\begin{equation}
\label{single photon field}
     \ket{1_\xi}
        =\int_{-\infty}^\infty \xi(\omega)\hat b^*(\omega) d\omega \ket{0}, 
\end{equation}
where $\hat b^*(\omega)$ is the Fourier transform of the white noise creation 
operator $\hat b^*(t)$, 
and $\xi(\omega)$ is the frequency domain shape function satisfying 
$\int_{-\infty}^\infty|\xi(\omega)|^2 d\omega = 1$ \cite{Gheri 1998}. 
%which differs from the single-mode single photon state 
%(i.e. $\ket{1}=\hat a^*\ket{0}$). 
%Rather $\ket{1_\xi}$ is a state of a continuous-mode pulsed light field; 
%$\hat b^*(f)$ is the Fourier transform of the white noise operator $\hat b^*(t)$, 
%and $\xi(f)$ is the pulse shape function satisfying 
%$\int_{-\infty}^\infty|\xi(f)|^2 df = 1$ \cite{Gheri 1998}. 

If $\ket{1_\xi}$ is taken as an input field state for a passive system that 
initially set to the ground state, then, in the long time limit the system 
returns to the ground state and the output is a single photon field state 
with pulse shape $\xi'(\omega)=\Xi(-i\omega)\xi(\omega)$ \cite{Guofeng2013}.
That is, as in the coherent input case, the output field state is completely 
characterized by the transfer function as follows:
\[
     \ket{1_{\xi'}}_{\rm out}
        =\int_{-\infty}^\infty \Xi(-i\omega)\xi(\omega)
              \hat b^*(\omega)d\omega \ket{0}. 
\]
We now suppose that the input pulse shape is enough broaden and so is 
confined around a fixed frequency $\omega$, thereby we denote 
$\ket{1_\xi}=\ket{1;\omega}$. 
Then, the output field state is given by 
$\ket{1; \omega}_{\rm out}= \Xi(-i\omega)\ket{1; \omega}$.
%
%\[
%    \ket{1_{\xi'}}_{\rm out}
%        = \Xi(i\omega) \int_{-\infty}^\infty \xi(f)\hat b^*(f) df \ket{0}
%        = \Xi(i\omega)\ket{1_\xi}.
%\]
%
The $n$-photon field state is defined in a similar way by \cite{Baragiola}:
\[
     \ket{n_\xi}
        = \frac{1}{\sqrt{n!}}
             \Big[ \int_{-\infty}^\infty \xi(\omega)\hat b^*(\omega) d\omega \Big]^n
             \ket{0}.
\]
As above, if the input for a linear passive system is a 
$n$-photon field state with its pulse shape confined at around $\omega$, 
then the output is given by 
$\ket{n; \omega}_{\rm out}=\Xi(-i\omega)^n\ket{n; \omega}$.

\end{document}